\documentclass[sigconf]{aamas}

\usepackage{balance} % for balancing columns on the final page

\doi{JAGQ1479}

\makeatletter
\gdef\@copyrightpermission{
  \begin{minipage}{0.2\columnwidth}
   \href{https://creativecommons.org/licenses/by/4.0/}{\includegraphics[width=0.90\textwidth]{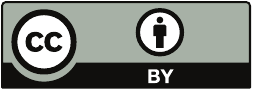}}
  \end{minipage}\hfill
  \begin{minipage}{0.8\columnwidth}
   \href{https://creativecommons.org/licenses/by/4.0/}{This work is licensed under a Creative Commons Attribution International 4.0 License.}
  \end{minipage}
  \vspace{5pt}
}
\makeatother

\setcopyright{ifaamas}
\acmConference[AAMAS '26]{Proc.\@ of the 25th International Conference
on Autonomous Agents and Multiagent Systems (AAMAS 2026)}{May 25 -- 29, 2026}
{Paphos, Cyprus}{C.~Amato, L.~Dennis, V.~Mascardi, J.~Thangarajah (eds.)}
\copyrightyear{2026}
\acmYear{2026}
\acmDOI{}
\acmPrice{}
\acmISBN{}

\title[Resilient Strategies for Stochastic Systems]{Resilient Strategies for Stochastic Systems:\\How Much Does It Take to Break a Winning Strategy?}

\subtitle{AAAI Track}

\author{Kush Grover}
\affiliation{
  \institution{Fondazione Bruno Kessler}
  \country{Italy}}

\author{Markel Zubia}
\affiliation{
  \institution{Ruhr University Bochum}
  \country{Germany}}

\author{Debraj Chakraborty}
\affiliation{
  \institution{Nanyang Technological University, Singapore}
  \country{}}

\author{Muqsit Azeem}
\affiliation{
  \institution{ 
	Technical University of Munich \&
	University of Konstanz
   }
  \country{Germany}}

\author{Nils Jansen}
\affiliation{
  \institution{Ruhr University Bochum
  \& Radboud University Nijmegen}
  \country{Germany}\\
  }

\author{Jan K\v{r}etinsk\'y}
\affiliation{
  \institution{Masaryk University}
  \country{Czech Republic}}

\begin{abstract}
We study the problem of resilient strategies in the presence of uncertainty. 
Resilient strategies enable an agent to make decisions that are robust against disturbances.
In particular, we are interested in those disturbances that are able to flip a decision made by the agent. 
Such a disturbance may, for instance, occur when the intended action of the agent cannot be executed due to a malfunction of an actuator in the environment.
In this work, we introduce the concept of resilience in the stochastic setting and present a comprehensive set of fundamental problems.
Specifically, we discuss such problems for Markov decision processes with reachability and safety objectives, which also smoothly extend to stochastic games.
To account for the stochastic setting, we provide various ways of aggregating the amounts of disturbances that may have occurred, for instance, in expectation or in the worst case. Moreover, to reason about infinite disturbances, we use quantitative measures, like their frequency of occurrence.
 
\end{abstract}


\usepackage{algorithm}
\usepackage{newfloat}
\usepackage{listings}

\usepackage{amsthm, amsfonts, mathtools, thmtools}
\newtheorem{definition}{Definition}
\newtheorem{example}{Example}

\newtheorem{lemma}{Lemma}

\newtheorem{corollary}{Corollary}
\newtheorem{remark}{Remark}

\usepackage{amsmath}
\usepackage{bm}
\usepackage{booktabs}
\usepackage{proof}
\usepackage{tikz}
\usetikzlibrary{arrows}
\usetikzlibrary{arrows.meta}
\usepackage{csquotes}
\usepackage{graphics}
\usepackage{stmaryrd}
\usepackage{soul}
\usepackage{array}
\usepackage{enumitem}
\usepackage{dsfont}
\usepackage{tcolorbox}

\usepackage{ellipsis, mparhack, ragged2e} % bugfixing
\usepackage[l2tabu, orthodox]{nag} % avoid typical LaTeX errors

\usepackage{multirow}
\usepackage{xparse}
\usepackage{mathtools}
\usepackage{environ}
\usepackage{lscape}

\usepackage{lipsum}

\usepackage{subcaption}
\usepackage{capt-of}

\usepackage[acronym]{glossaries}
\usepackage{algpseudocode}

\usepackage{marginfix}
\usepackage{xifthen}

\usepackage{longtable} % table over several pages

\usepackage{afterpage}

\usepackage{xfrac}

\usetikzlibrary{automata,positioning,fit,positioning,arrows}
\tikzset{
	state/.style={
		circle,
		draw=black, very thick,
		minimum height=2em,
		inner sep=2pt,
		text centered,
	},
	squarednode/.style={rectangle, draw=black, very thick, minimum size=2em},
	diamondnode/.style={diamond, draw=red!60, fill=red!5, very thick, minimum size=7mm},
	trianglenode/.style={regular polygon,regular polygon sides=3, draw=blue!60, fill=blue!5, very thick, minimum s ize=5mm},
	smallcircle/.style={circle, draw=black, fill=black!5, very thick, scale=0.5}
}

\usepackage[capitalise]{cleveref}

\usepackage{placeins} % for FloatBarrier

\usepackage{pgfplots}
\usepgfplotslibrary{fillbetween}
\pgfplotsset{compat=newest}

\algnewcommand\algorithmicswitch{\textbf{switch}}
\algnewcommand\algorithmiccase{\textbf{case}}
\algnewcommand\algorithmicassert{\texttt{assert}}
\algnewcommand\Assert[1]{\State \algorithmicassert(#1)}%
\algdef{SE}[SWITCH]{Switch}{EndSwitch}[1]{\algorithmicswitch\ #1\ \algorithmicdo}{\algorithmicend\ \algorithmicswitch}%
\algdef{SE}[CASE]{Case}{EndCase}[1]{\algorithmiccase\ #1}{\algorithmicend\ \algorithmiccase}%
\algtext*{EndSwitch}%
\algtext*{EndCase}%

\newcommand{\para}{\paragraph}

\definecolor{blizzardblue}{rgb}{0.67, 0.9, 0.93}
\definecolor{lightcyan}{rgb}{0.9, 1.0, 1.0}

\newcommand{\distribution}{\mathsf{Dist}}

\newcommand{\mdpstates}{S}
\newcommand{\mdpinitstates}{s_0}

\newcommand{\game}{\mathcal{G}}
\newcommand{\gamestates}{S}
\newcommand{\gamestatesone}{S_1}
\newcommand{\gamestatestwo}{S_2}

\newcommand{\gameactions}{A}
\newcommand{\gameavailactions}{Av}
\newcommand{\gametransitions}{T}
\newcommand{\gamedisturbanceactions}{A^D}
\newcommand{\gameavaildisturbanceactions}{Av^D}
\newcommand{\gamedisturbancetransitions}{T^D}
\newcommand{\gamepath}{\rho}

\newcommand{\newgame}{\tilde{\mathcal{G}}}
\newcommand{\newgamestates}{\tilde{S}}
\newcommand{\newgamestatesone}{\tilde{S_1}}
\newcommand{\newgamestatestwo}{\tilde{S_2}}
\newcommand{\newgameactions}{\tilde{A}}
\newcommand{\newgameavailactions}{\tilde{Av}}
\newcommand{\newgametransitions}{\tilde{T}}

\newcommand{\strategyone}{\pi}
\newcommand{\strategytwo}{\sigma}
\newcommand{\disturbancestrategy}{{\delta}}

\newcommand{\objective}{\phi}
\newcommand{\targetstates}{G}

\newcommand{\breakpoint}{\mathcal{B}}
\newcommand{\expectedbreakpoint}{\mathcal{EB}}
\newcommand{\expected}{\mathbb{E}}
\newcommand{\prob}{\mathbb{P}}
\newcommand{\disturbances}{\mathcal{D}^{T}}
\newcommand{\longrundisturbances}{\mathcal{D}^{F}}

\DeclareGraphicsExtensions{.pdf, .png, .jpg}

\usepackage{etoolbox}
\makeatletter
\AfterEndEnvironment{definition}{\@doendpe}
\AfterEndEnvironment{theorem}{\@doendpe}
\AfterEndEnvironment{corollary}{\@doendpe}
\AfterEndEnvironment{lemma}{\@doendpe}
\AfterEndEnvironment{proposition}{\@doendpe}
\AfterEndEnvironment{example}{\@doendpe}

\AfterEndEnvironment{definition}{\everypar{\setbox\z@\lastbox\everypar{}}}
\AfterEndEnvironment{theorem}{\everypar{\setbox\z@\lastbox\everypar{}}}
\AfterEndEnvironment{corollary}{\everypar{\setbox\z@\lastbox\everypar{}}}
\AfterEndEnvironment{lemma}{\everypar{\setbox\z@\lastbox\everypar{}}}
\AfterEndEnvironment{proposition}{\everypar{\setbox\z@\lastbox\everypar{}}}
\AfterEndEnvironment{example}{\everypar{\setbox\z@\lastbox\everypar{}}}
\makeatother

\RequirePackage{marvosym}

\newcommand{\Naturals}{\mathbb{N}}

\DeclareDocumentCommand{\post}{D<>{} O{} D(){}}{\mathsf{Post}_{#1}^{#2}\ifthenelse{\isempty{#3}}{}{(#3)}}

\newcommand{\polytime}{\mathbf{P}}
\newcommand{\NP}{\mathbf{NP}}

\newcommand{\PSPACE}{\mathbf{PSPACE}}

\newcommand{\N}{\mathbb{N}}

\makeatletter
\newcommand{\algcolor}[2]{%
	\hskip-\ALG@thistlm\colorbox{#1}{\parbox{\dimexpr\linewidth-2\fboxsep}{\hskip\ALG@thistlm\relax #2}}%
}

\makeatother

\renewcommand{\algorithmiccomment}[1]{\bgroup\hfill//~#1\egroup}

\definecolor{algo-teal}{rgb}{0.84, 0.85, 0.97}
\definecolor{algo-lime}{rgb}{0.898, 1, 0.53}
\definecolor{algo-pink}{rgb}{1, 0.827, 0.878}
\definecolor{algo-yellow}{rgb}{1, 0.98, 0.74}
\definecolor{persianindigo}{rgb}{0.2, 0.07, 0.48}
\definecolor{lincolngreen}{rgb}{0.11, 0.35, 0.02}
\definecolor{mulberry}{rgb}{0.77, 0.29, 0.55}
\definecolor{goldenpoppy}{rgb}{0.99, 0.76, 0.0}
\definecolor{our-red}{rgb}{1,0.7,0.7}
\definecolor{our-green}{rgb}{0.7,1,0.7}

\newtoggle{showplots}
\toggletrue{showplots}

\newenvironment{new}{}{}
\newenvironment{added}{}{}

\begin{document}

%%% The following commands remove the headers in your paper. For final 
%%% papers, these will be inserted during the pagination process.

\pagestyle{fancy}
\fancyhead{}

%%% The next command prints the information defined in the preamble.

\maketitle 

%%%%%%%%%%%%%%%%%%%%%%%%%%%%%%%%%%%%%%%%%%%%%%%%%%%%%%%%%%%%%%%%%%%%%%%%

\section{Introduction and Motivation}

In many areas such as machine learning, robotics, automated planning, and game theory, a notion of optimality is adopted to argue about the best possible behavior in a given environment.
While optimality captures the best possible performance an agent may achieve in a fixed, well-defined environment, real-world settings are rarely stable or predictable. Then, resilience becomes a more meaningful objective: it reflects the agent's ability to maintain desirable behavior despite changes or adversarial disturbances. 
The latter includes notions such as fault tolerance of algorithms, robustness to disturbances in control, or trembling hands equilibria in game theory.
As \citet{vardi2020efficiency} argues, this notion is spread not only over computer science, but also economics, evolution, or other dynamic systems; and while considerable effort has been spent in the area, still 
\emph{``We must recognize the trade-off between efficiency and resilience. It is time to develop the discipline of resilient algorithms.''}

\para{Resilient Strategies.}
This paper focuses on resilience of \emph{strategies} (a.k.a.\ policies, schedulers, or controllers, depending on the context) in terms of the number of decisions that need to be subverted, which we refer to as \emph{disturbances}, to violate the property otherwise satisfied by the original strategy.
From the control perspective, such disturbances may represent actuator faults, control noise, or environmental perturbations that prevent the agent from executing its intended action.
For instance, consider an autonomous drone navigating a grid world toward a target (\Cref{fig:drone}): the strategy corresponding to the red path would reach the target in the ideal situation. 
However, it passes close to the trees, making it vulnerable to collisions from disturbances (due to wind). 
A resilient strategy, in contrast, would follow the green path which maintains a buffer from the trees. This leads to the target with probability $1$, but it can also handle a disturbance.
Understanding how sensitive a strategy is to such disturbances is essential for deploying agents in uncertain or adversarial environments.
\begin{figure}[ht]
	\centering
	\includegraphics[width=78pt,height=64pt]{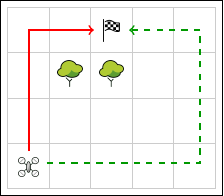}
	\caption{Drone navigation task with wind disturbances.}
	\label{fig:drone}
\end{figure}

From a theoretical perspective, analyzing resilience reveals which regions of the state space are critical to maintaining correct behavior. 
That is, it helps identify the brittle points of the agent's strategy where a few disturbances can lead to failure. 
This information is relevant in the context of \emph{explaining} its key decisions \cite{ashok2020dtcontrol}, in the context of \emph{strategy repair}~\cite{DBLP:conf/nfm/PathakAJTK15}, and in the modern on-the-fly automated synthesis approaches \cite{kvretinsky2023guessing}.
Notably, standard modeling approaches such as incorporating noise probabilistically or assuming a fully adversarial environment fall short of capturing this nuance. 
Probabilistic models mask the fragility of decisions by averaging over expected behaviors, 
while adversarial models are overly pessimistic and can rule out efficient paths unnecessarily. 
In contrast, our disturbance-based view offers a finer robustness metric by asking how many deviations are needed to cause failure.

\begin{table*}[t]
	\caption{Overview of strategy evaluation and optimal strategy computation under different semantics and breaking point types.}
	\centering
	\resizebox{\textwidth}{!}{
	\begin{tabular}{l|l|c|c|c|c}
		\toprule
		\multirow{2}{*}{\textbf{Semantics}} & 
		\multirow{2}{*}{\shortstack{\textbf{Breaking} \\ \textbf{Point Type}}} & 
		\multicolumn{2}{c|}{\textbf{Strategy Evaluation}} & 
		\multicolumn{2}{c}{\textbf{Strategy Synthesis}} \\
		\cmidrule(lr){3-4} \cmidrule(lr){5-6}
		& & \textbf{Safety} & \textbf{Reachability} & \textbf{Safety} & \textbf{Reachability} \\
		\midrule
		
		\multirow{2}{*}{Expected} 
		& Transient 
		& SSP for MDP \textbf{(P)} 
		& SSP for MDP \textbf{(P)} 
		& SSP for SG 
		& SSP for SG \\
		
		& Frequency 
		& Cannot happen
		& Collapse MECs $+$ SSP for MDP \textbf{(P)} 
		& Cannot happen 
		& Collapse MECs $+$ SSP for SG \\
		
		\midrule
		
		\multirow{2}{*}{Worst-case} 
		& Transient 
		& Iterative LP \textbf{(PSPACE)}
		& Iterative LP \textbf{(PSPACE)}
		& Iterative QP \textbf{(PSPACE)} 
		& Iterative QP \textbf{(PSPACE)} \\
		
		& Frequency 
		& Always $0$
		& Collapse MECs $+$ Worst-case analysis \textbf{(P)} 
		& Always 0 
		& Collapse MECs $+$ Worst-case analysis \textbf{(NP)}  \\
		
		\bottomrule
	\end{tabular}
	}
	\label{tab:breaking_point}
\end{table*}

\para{Resilience in Stochastic Systems.}
Previous work has investigated the resilience of strategies in (non-quantitative) graph games \cite{DBLP:journals/acta/NeiderWZ20}.

While the landscape is rather bland there, it becomes very vibrant in the stochastic context.
Indeed, in the former setting, 
(i)~quantifying resilience by the number of disturbances (decisions to be subverted) during a play is a straightforward choice, 
(ii)~it results in an integer bounded by the size of the state space or infinity,
(iii)~algorithmically, it boils down to iterative, graph-search-based procedures that identify regions from which an agent can enforce reaching a goal regardless of disturbances.
In stark contrast, for Markov decision processes and stochastic games,
(i)~resilience purely in terms of expected disturbance counts may not be adequate, and we also analyze worst cases happening with positive probability;
(ii)~the number of disturbances may now be larger than the state space size and is, in fact, unbounded; in such cases, we refine the resilience measure by considering the frequency of disturbances rather than their total count;
(iii)~algorithmically, only some cases can be easily reduced to simple graph search or mean-payoff computations.

\begin{new}
\para{Assumptions on Model Availability.}
Our work relies on the standard assumption in the model-based planning paradigm that the underlying model is available in explicit form or can be obtained through established learning methods. 
This convention is well documented in the model-checking literature; see, for example~\citep{DBLP:books/daglib/0020348}, where verification algorithms are defined relative to a given system model. 
When the model is not provided a priori but is learnable from data, statistical model checking offers a widely used pathway for reconstructing or approximating the relevant stochastic behavior~\citep{agha2018survey,ashok2019pac}. 
Once such an approximation is obtained, the resulting model fits within the standard verification framework and our method applies directly. 
\end{new}

\begin{new}
	\para{Disturbances with Different Costs.}
	In this work, we focus on symmetric disturbances, i.e. all disturbances have the same cost, for two reasons:
	(i)~Prior work on resilience in non-stochastic settings \cite{cdc2016,DBLP:journals/acta/NeiderWZ20} has considered only symmetric disturbances, and we aim to extend that line of work to stochastic settings.
	(ii)~From the practical perspective, assigning precise costs to disturbances is often challenging, as it requires detailed domain knowledge.
	On the other hand, symmetric disturbances appropriately model several realistic scenarios, for example network routing problems \cite{paxson1997end,hopps2000rfc2992}.
	Furthermore, our framework can naturally extend to asymmetric disturbance costs by assigning different weights to disturbance transitions.
	This would change the resilience metric from ``how many disturbances'' to ``total disturbance cost''.
	This would require adapting some algorithms presented here, while the rest would remain unchanged, and we leave this for future work.
\end{new}

\para{Summary of Our Contribution:}
\begin{itemize}
	\item We extend the notion of strategy resilience to the stochastic setting with safety and reachability objectives.
	To that end, we consider expected and positive-measure worst cases.
	When infinitely many disturbances are required to break a strategy, we refine the notion by computing their frequency.
	\item We provide algorithms for computing the resilience of a given strategy (for each definition of resilience introduced) and to compute optimally resilient strategies. 
	From the perspective of the efficiency-resilience trade-off, these are satisfying strategies (satisfy the functional property with at least a given probability threshold) with optimal resilience (requiring the most disturbances to break them, i.e., decrease the probability below the threshold).
\end{itemize}

\begin{new}
The resilience of a strategy is measured by calculating its \emph{breaking point}, defined as the maximum number of disturbances that are necessary to break it.
\end{new}
\cref{tab:breaking_point} provides an overview of the algorithms to compute the breaking point of a pure memoryless strategy in the worst and expected cases, and for finding a strategy with the maximum breaking point.
Definitions of terms used in the tables are provided in Section \ref{sec:prelims}.

\para{Related Work.}
The concept of system robustness against errors has been extensively studied in various contexts. 
One common approach is to model uncertainties in estimated probabilities using an uncertainty set within which the true probability resides. 
To ensure absolute safety, a worst-case analysis is often employed, yielding results that are resilient to such disturbances~\cite{reactive-synthesis}. 
\begin{new}
    Another approach is to perform a sensitivity analysis over the uncertainty sets to measure how resilient the system is with respect to different variables~\cite{sensitivity-analysis-mdp}.
\end{new}
For a comprehensive comparison of different notions of resilience and robustness, we refer the reader to~\cite{cdc2016, DBLP:journals/acta/NeiderWZ20}.

\begin{new}
Apart from uncertainty sets, notions of differential privacy and deception for MDPs have also been explored, which utilize and optimize for different measures of detection and resilience~\cite{differential-privacy-1, differential-privacy-2, deception-mdp}.    
\end{new}
Another way to model disturbances is by treating them as random events with small probabilities, often called \emph{trembling hand}~\cite{Marchesi_Gatti_2021, ijcai2024p402}.
However, this might not always be suitable, and defining an accurate stochastic error model may be challenging, also argued in~\cite{DBLP:journals/acta/NeiderWZ20}.

Therefore, the idea of strategies that are resilient to unmodeled intermittent disturbances was first introduced in~\cite{cdc2016} within the context of safety games. 
This concept was later extended to prefix-independent winning conditions, including parity objectives, in~\cite{DBLP:journals/acta/NeiderWZ20}. 
It was demonstrated that computing optimally resilient strategies for parity conditions incurs only a polynomial overhead compared to solving traditional parity games. 
These ideas were further applied in~\cite{cdc2020} to find resilient controllers for a continuous dynamical system. 
Additionally, the notion of resilient strategies was extended to infinite arenas, particularly pushdown graphs, in~\cite{DBLP:conf/mfcs/NeiderT020}.

However, none of these works consider the potential benefits of leveraging partial knowledge of stochasticity within the model. 
To our knowledge, this work is the first to study unmodeled intermittent disturbances in the context of stochastic systems such as Markov decision processes and stochastic games. 
Moreover, prior works do not investigate resilience in terms of the frequency of disturbances required to compromise a strategy, particularly in scenarios where infinite disturbances can be handled by the controller.

\section{Preliminaries}\label{sec:prelims}

For a set $A$, we denote its power set by $\mathit{Pow}(A)$. 
% \nj{set $\mathit{Pow}(A)$}
We use Von Neumann ordinals, with $\langle n \rangle = \{0, \dots, n-1\}$ for all $n \in \mathbb{N}$.
The set of all (discrete) probability distributions on the set $S$ is denoted by $\distribution(S)$.
We denote by $S^*$ the set of all \emph{finite strings} over $S$.

\subsection{Stochastic Games}
\begin{definition}
	A {\new{(2-player) }} \emph{Stochastic Game} (SG) is a tuple $\game = (\gamestates, \gamestatesone, \gamestatestwo, \gameactions, \gameavailactions, \gametransitions)$ where 
	$\gamestates$ is a finite set of states partitioned into Player~1 states $\gamestatesone$ and Player~2 states $\gamestatestwo$;
	$\gameactions$ is a finite set of actions;
	$\gameavailactions:\gamestates\rightarrow \mathit{Pow}(\gameactions)$ is a total function which maps the set of available actions to each state, and
	$\gametransitions: \gamestates\times \gameactions\rightharpoonup \distribution(\mdpstates)$ is a partial transition function, which, given a state and an available action, returns a probability distribution over the successor states.
\end{definition}
We also define an initial state, $s_0 \in \gamestates$ from which the play starts.
With a slight abuse of notation, we use $T(s,a,s')$ to denote $T(s,a)(s')$.
For this paper, we assume that the stochastic game is turn-based, i.e., in each state, only one player can move.

\paragraph{Paths or Runs}
The set of successor states for $s\in \gamestates$ that can be reached by taking the action $a$ is $\post[\game](s,a)=\{s'\mid \gametransitions(s,a,s')>0\}$. 
A \emph{finite path} (or \emph{finite run}) $\varrho = s_0a_0s_1\ldots s_i$ of length $i\ge 0$ is a sequence of states and actions such that for all $t\in[0,i-1]$, $a_t\in \gameavailactions(s_t)$ and $s_{t+1}\in \post[\game](s_t,a_t)$.
We can define infinite paths (or infinite runs) $\gamepath =s_0a_0s_1a_1s_2\ldots$ analogously.

\paragraph{Strategies}
A Player~1 strategy is defined as a function $\strategyone:(\gamestates\times \gameactions)^*\times \gamestatesone\rightarrow \distribution(\gameactions)$ mapping the history to a distribution over available actions.
A strategy is \emph{pure} if the distribution is always a Dirac delta function, and it is \emph{mixed} otherwise.  
Player~2 strategies have analogous definitions.

\paragraph{Markov Decision Processes (MDP) and Markov Chains (MC)} MDPs are a special case of stochastic games where the states of one player are empty i.e. $\gamestatestwo = \emptyset$.
If the strategy of Player~1, $\strategyone$, is fixed in the SG $\game$, it induces an MDP $\game^{\strategyone}$.
An MC is an MDP with $|\gameavailactions(s)|=1$ for all $s\in\mdpstates$. 
Fixing a strategy $\strategytwo$ on an MDP $\game^{\strategyone}$ induces a Markov chain $\game^{\strategyone, \strategytwo}$. 
An MC $M$ and an initial state $s_0$ define a unique probability measure $\prob_{M,s_0}$ over infinite paths~\cite{puterman2014markov}.
For any random variable $X$ defined over the infinite paths of MC $M$, its expected value with respect to $\prob_{M,s_0}$ is  $\expected_{M,s_0}[X]$.

\paragraph{End Component (EC)}
	A set of states and actions is an \emph{end component} if the play never leaves that set and it is a \emph{maximal end component} (MEC) if it cannot be extended by adding more states and actions.
	The states in an MDP can be partitioned into MECs, and every play is guaranteed to eventually enter a MEC.
	We refer to standard literature~\cite{puterman2014markov} for the formal definitions.
	We use $\mathsf{MEC}(M)$ to denote the set of maximal end components of $M$.
	This notion of MECs can be extended to SGs.

\paragraph{Rewards or Costs.} We model the costs associated with each state-action pair using a function
$C: \gamestates \times \gameactions \rightarrow \Naturals$.
Given a Markov chain (MC) $M$, let $C_i$ be a random variable that, for an infinite path $\rho = s_0 a_0 s_1 a_1 \dots$, returns
$C_i(\rho) = C(s_i, a_i)$; i.e., the cost (or reward) incurred at the $i$-th step of the path.
The total cost for an infinite path $\rho$ is defined as
$
C_\rho := \sum_{j=0}^{\infty} C(s_j, a_j)
$.
The \emph{expected total cost} for $M$ starting from state $s_0$ is defined as
$
\mathsf{TR}(M, s_0, C) := \sum_{k=0}^\infty \expected_{M, s_0}(C_k) = \expected_{M, s_0}(C_\rho)
$.
The $k$-step average cost from state $s_0$ in $M$ is defined as
$
v_k(s_0) := \expected_{M, s_0} \left( \frac{1}{k+1} \sum_{j=0}^{k} C_j \right)
$.
The \emph{expected mean payoff} for $M$ from state $s_0$ is defined as
$
\mathsf{MP}(M, s_0, C) := \liminf_{k \rightarrow \infty} v_k(s_0)
$.
The mean payoff for a path $\rho$ is defined as $\mathsf{MP}(\rho) := \liminf_{k \rightarrow \infty} \frac{1}{k+1} \sum_{j=0}^{k} C(s_j, a_j)$.

\paragraph{Specifications.}
We use the standard temporal operators $\square$ (globally) and $\lozenge$ (eventually). 
In an MC $M$ with initial state $s_0$, we 
let $\mathbb{P}_{\game,s_0}(\square\neg B)$ denote the probability of \emph{never} visiting $B \subseteq S$, 
and $\mathbb{P}_{\game,s_0}(\lozenge G)$ denote the probability of \emph{eventually} reaching $G \subseteq S$.
We focus on two standard quantitative specifications and their negations: \emph{safety}, where 
$\phi_{safety} := \prob_{\game,s_0}(\square \neg B) > p$; and \emph{reachability}, where 
$\phi_{reach} := \prob_{\game, s_0}(\lozenge G) > q$.

The properties $\phi_{safety}$ and $\phi_{reach}$ could equivalently be stated with $\geq$; the analysis carries over with only minor adaptations. 
Under expected semantics, the same computation yields \emph{resilience} rather than the \emph{breaking point}, while the worst-case semantics remains unchanged.
W.l.o.g., we assume that all states in $G$ or $B$ are sink states with only a self loop.

\paragraph{Algorithms to Solve MDPs and SGs.}
For solving \emph{reachability} and \emph{safety} problems in MDPs and SGs, key algorithms includes
Value Iteration (VI), Linear (LP) and Quadratic Programming (QP),
and Policy Iteration (PI).
We refer the reader to standard texts~\cite{puterman2014markov, condon1990algorithms} for details of these algorithms. 
For solving \emph{stochastic shortest path}~(SSP) problems, that ask to minimize the cost of reaching a target state, one can use LP for MDPs and VI/PI for SGs.

\section{Stochastic Games with Disturbances}
\label{sec:sgd}
In this section, we define stochastic games with disturbances.
In these systems, disturbances may occur at runtime and override the decisions of Player~1. 
For example, as seen in \cref{fig:drone}, if the drone wants to move upwards, wind can disturb its actions and push it to the right, completely changing its decision.
\begin{definition}
	A \emph{Stochastic Game with Disturbances} (SGD) is a tuple $\game = (\gamestates, \gamestatesone, \gamestatestwo, \gameactions, \gameavailactions,$ $\gametransitions, \gamedisturbanceactions, \gameavaildisturbanceactions, \gamedisturbancetransitions)$ where 
	$(\gamestates, \gamestatesone, \gamestatestwo, \gameactions, \gameavailactions, \gametransitions)$ is an SG, $\gamedisturbanceactions$ is a finite set of \emph{disturbance actions} (disjoint from $A$), $\gameavaildisturbanceactions : \gamestatesone\rightarrow Pow(\gamedisturbanceactions)$ is a function specifying the \emph{available disturbance actions} in Player~1 states, and 
	$\gamedisturbancetransitions: \gamestatesone\times \gamedisturbanceactions \rightharpoonup \distribution(\gamestates)$ defines a \emph{disturbance transition function}.
\end{definition}
We use $|\gamedisturbancetransitions|$ to denote the number of disturbance transitions in the game.
We can define an initial state and strategies for both players in the same way as for an SG.
However, since we have disturbance actions here, we also define a disturbance strategy.
A \emph{disturbance strategy} is defined as $\disturbancestrategy: (S\times A)^* \times S_1 \rightarrow Dist(\gamedisturbanceactions \cup \{\bot\})$.
Here, $\bot$ represents that no disturbance action is taken.
A run of the game is a sequence of states and actions $\gamepath = s_0 a_0 s_1 a_1 \dots$ where $\forall i, s_i \in \gamestates$, $a_i \in \gameavailactions(s_i)\cup \gameavaildisturbanceactions(s_i)$, and $s_{i+1} \in \gametransitions(s_i, a_i)$.
For a run $\gamepath = s_0 a_0 s_1 a_1,\dots$, we define the total number of disturbances as 
$\disturbances(\gamepath):= \big|\{a_i\mid a_i\in\gameavaildisturbanceactions(s_i)\}\big|$.
This number may not be finite, and in that case, we can define the frequency of disturbances as 
$\longrundisturbances(\gamepath):=\liminf_{k\rightarrow\infty} \frac{1}{k}\big|\disturbances(\gamepath[0,k])\big|$,
where $\gamepath[0,k]$ denotes the $k$-length prefix of an infinite run $\gamepath$.
We need $\liminf$ here because the limit in general might not exist.

\paragraph{Modeling Disturbances as Costs}
\label{sec:cost}
We can model the disturbances as costs in an SGD by letting $C(s, a) = 1$ if 
$s \in \gamestatesone, a \in \gamedisturbanceactions$, and letting $C(s, a) = 0$
otherwise.
\begin{remark}
	\label{rem:disturbance_cost_equivalence}
	Note that for a run $\gamepath$, the number of disturbances and the frequency of disturbances can be seen as the total cost and the mean payoff of the run respectively, i.e.
	\[ \disturbances(\gamepath) = C_\gamepath \text{ and } \longrundisturbances(\gamepath) = \mathsf{MP}(\gamepath)\,. \]
\end{remark}

\subsection{Resilience, Breaking Point, and Their Different Semantics}
Since disturbances are adversarial, we can assume that they are performed by Player~2.
The \emph{resilience} of a strategy is defined as the maximum number of disturbances that the strategy can withstand while satisfying a given specification.
However, this maximum might not exist, for instance if it can withstand $<2.5$ disturbances on average but cannot withstand 2.5 disturbances. 
Therefore, a more suitable, well-defined notion would be the \emph{breaking point} of a strategy, which refers to the minimum number of disturbances required to break it.
In this setting, \emph{an \textbf{optimally resilient strategy} is a strategy that has the greatest breaking point.}

There are various natural ways to quantify how many disturbances are required to break a strategy.
One approach is to consider the expected number of disturbances needed to break said strategy.
Another approach is to determine the maximum number of disturbances required over all possible paths, which we refer to as the worst-case measure throughout this paper.
The choice between these measures depends on the application context.
Expected resilience captures average-case behavior and is appropriate when disturbances occur frequently.
However, it may obscure rare but critical scenarios in which only a few disturbances suffice to break the strategy.
For robustness guarantees, it is often more informative to know how many disturbances the strategy can withstand in the worst case.

\begin{new}
\begin{figure}[h]
	\tiny
	\centering

	\begin{tikzpicture}[scale=0.5]
		
		\node[state, initial, initial text=] at (0,0)     (0)   {$s_0$};
		\node[state] at (0,-2)     (1)   {$s_1$};
		\node[state, fill=our-green] at (-3,-2)     (2)   {G};
		%		\node[state, fill=our-green] at (4,0)    (3)   {$s_3$};
		\node[state, fill=our-red] at (3,-2)    (3)   {B};
		\node[smallcircle] at (1.5,-0.5)    (4)   {};
		\node[smallcircle] at (1.5,-3.5)     (5)   {};
		
		\draw
		(0) edge[->, above] node {$a$} (2)
		(0) edge[-, above, dashed] node {$d$} (4)
		(4) edge[->, left, dashed] node {$0.5$} (1)
		(4) edge[->, right, dashed] node {$0.5$} (3)
		(1) edge[->, above] node {$a$} (2)
		(1) edge[-, above, dashed] node {$d$} (5)
		(5) edge[->, below, dashed] node {$0.5$} (2)
		(5) edge[->, right, dashed] node {$0.5$} (3)
		
		;
	\end{tikzpicture}
	\caption{For this SGD, if the objective is to reach $G$ with probability $> 0.4$, then the worst-case breaking point is 2, as disturbing just once will not break the policy, and the expected breaking point is 1.1 via an adversary that always disturbs in $s_0$ and with a probability of 0.2 in $s_1$.}
	\label{fig:worst_vs_average_breaking_point}
\end{figure}
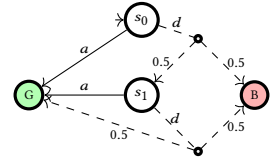
\end{new}

\subsection{Induced MDP Under a Player~1 Strategy}
\label{sec:induced-MDP}
Given an SGD $\game$ and a memoryless Player~1 strategy $\strategyone$, we use $M_{\strategyone}$ to denote the \emph{induced MDP} under $\strategyone$, described next. 
For every state $s \in \gamestatesone$, we remove all Player~1 actions except for the one picked by $\strategyone(s)$, while the disturbance actions remain unchanged.
Now, we can think of all states as Player~2 states.

\begin{restatable}{lemma}{theoremInducedMDPEquivalence}
	\label{lem:induced_mdp_equivalence} 
	For a memoryless Player~1 strategy $\strategyone$, there exist a bijective function $h$ that maps pairs $(\strategytwo, \delta)$ of Player~2 and disturbance strategies in $\game$ to a strategy $\mu$ in $M_\strategyone$ such that
	\[
	\forall \rho \in (S \times A)^* \times S, \quad \prob^{\pi, \sigma, \delta}_{\game, s_0}(\rho) = \prob^{\mu}_{M_\strategyone, s_0}(\rho)
	.\]
\end{restatable}

This lemma implies that the pair of strategies $(\strategytwo, \delta)$ for $\game$ is equivalent to strategy $\mu$ for MDP $M_\strategyone$.

For a Player~1 strategy $\strategyone$, a Player~2 strategy $\strategytwo$, a disturbance strategy $\disturbancestrategy$ and an initial state $\mdpinitstates$, we can define the induced MC as follows: 
In Player~2 states, the action is selected by $\strategytwo$.
In Player~1 states, if $\disturbancestrategy(s) = \bot$ then we pick action $\pi(s)$, whereas we pick $\disturbancestrategy(s)$ otherwise.
We denote the probability measure in the induced MC as $\prob^{\pi, \sigma, \delta}_{\game, s_0}$.
For any random variable $X$ defined over infinite paths in $\game$, its expected value with respect to $\prob^{\pi, \sigma, \delta}_{\game, s_0}$ is  $\expected^{\pi, \sigma, \delta}_{\game, s_0}(X)$.

\subsection{From SGDs with Finite Disturbances to SGs}
\label{sec:unfolded-SG}
Here, we show that if the number of disturbances is bounded (by some $k\in \N$) on all runs, then we can transform an SGD $\game$ into an equivalent SG $\game^{\dagger k}$.
This is achieved by unfolding the state space of the SGD to account for the number of disturbances that have already occurred.
Specifically, we encode the remaining number of disturbances into the state space, i.e., the new states are of the form $(s, i)$, where $s$ is a state in $\game$ and $i \in \langle k + 1 \rangle$ represents the number of remaining disturbances. 
\cref{fig:gadget_unfold} illustrates this transformation. 

Intuitively, after Player~1 chooses an action in a state $s \in S_1$, a disturbance may or may not occur. 
To embed this extra step into a standard stochastic game, we create an intermediate state $(s,i,a)$, controlled by Player~2, that is reached whenever Player~1 plays action $a$ from $s$.
Player~2 then chooses either $\bot$ (no disturbance) or a disturbance action $d$. 
If $\bot$ is chosen, we proceed with the usual transition; if $d$ is chosen, the transition follows $\gamedisturbancetransitions$ and we decrement the disturbance counter~$i$. 
Hence, each disturbance possibility becomes a normal turn-based move of Player~2, making the disturbance mechanism explicit in the resulting unfolded game.
It can easily be shown that this transformation is sound by showing the equivalence of the two games.
\begin{added}
	\begin{restatable}{lemma}{theoremUnfoldedSGEquivalence}
		\label{thm:unfolded_sg_equivalence} 
		There exist bijective functions $f$ and $g$ where 
		$f$ maps triplets $(\pi, \sigma, \delta)$ of strategies in $\game$ to pairs $(\tilde \pi, \tilde \sigma)$ of strategies in $\game^{\dagger k}$, 
		and $g$ maps paths $\rho$ in $\game$ to paths $\tilde \rho$ in $\game^{\dagger k}$, such that
		$
		\forall \rho \in (S \times A)^* \times S, \quad \prob^{\pi, \sigma, \delta}_{\game, s_0}(\rho) = \prob^{\tilde \pi, \tilde \sigma}_{\game^{\dagger k}, (s_0, k)}(\tilde \rho)
		$,
		as long as $\delta$ disturbs at most $k$ times on each run.
	\end{restatable}
\end{added}

\begin{figure}[t]
	\scriptsize
	\centering
	\begin{tikzpicture}[scale=0.6]
		
		\node[state] at (0,0)     (0)   {$s$};
		\node[] at (-1,1) (1) {};
		\node[] at (-1,-1) (2) {};
		\node[] at (3,1) (3) {$s'$};
		\node[] at (3,-1) (4) {$s''$};
		\node[smallcircle] at (1.5,-0.5) (circle1) {};
		\node[smallcircle] at (1.5,0.5) (circle2) {};

		\node[state] at (5.4,0)     (5)   {$(s, i)$};
		\node[] at (4.5,1) (6) {};
		\node[] at (4.5,-1) (7) {};
		\node[squarednode] at (8.0,0)     (9)   {$(s, i, a)$};
		\node[] at (11.7,-1) (10) {$(s'', i-1)$};
		\node[] at (11.2,1) (11) {$(s', i)$};
		\node[smallcircle] at (9.7,-0.5) (circle3) {};
		\node[smallcircle] at (9.7,0.5) (circle4) {};

		\node[] at (3.9,0) {$\rightsquigarrow$};
		
		\draw
		(1) edge[->, dotted] node {} (0)
		(2) edge[->, dotted] node {} (0)
		(0) edge[-, above] node {$a$} (circle2)
		(0) edge[-, dotted, below] node {$d$} (circle1)
		(circle2) edge[->, above] node {$p_1$} (3)
		(circle1) edge[->, dotted, below] node {$p_2$} (4)

		(6) edge[->, dotted] node {} (5)
		(7) edge[->, dotted] node {} (5)
		(9) edge[-, below] node {$d$} (circle3)
		(circle3) edge[->, below] node {$p_2$} (10)
		(9) edge[-, above] node {$\bot$} (circle4)
		(circle4) edge[->, above] node {$p_1$} (11)
		(5) edge[->, above, pos=0.4] node {$a, 1$} (9)
		
		;
	\end{tikzpicture}
	\caption{A gadget in the unfolded stochastic game.}
	\label{fig:gadget_unfold}
\end{figure}
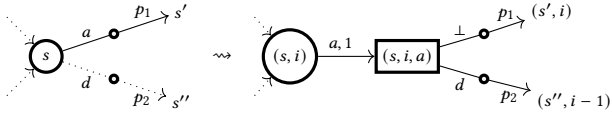

\section{Expected Breaking Point}
\label{sec:expected-breaking-point}

For a given Player~1 strategy $\strategyone$, we define the set of pairs of strategies that break $\strategyone$ as
$U_\strategyone:= \{(\strategytwo, \disturbancestrategy ) \mid \game^{\strategyone,\strategytwo,\disturbancestrategy}_{s_0}\models \neg\phi\}$.
We define the \emph{expected transient breaking point} of $\strategyone$ as
\[\expectedbreakpoint^T_{\game,\objective}(\strategyone):=
\inf_{(\strategytwo,\disturbancestrategy) \in U_\strategyone} \expected_{\game,s_0}^{\strategyone,\strategytwo,\disturbancestrategy}\big(\disturbances(\rho) \big).
\]
When set $U_\strategyone$ is empty, which means that the strategy cannot be broken, we refer to it as $\natural$. When there exist no finite values in $U_\strategyone$, we refer to it as $\omega$.
If breaking the strategy requires infinitely many disturbances on average, we can still ask how frequently disturbances are required to break the strategy.
For such cases, the \emph{expected frequency breaking point} is defined as
\[\expectedbreakpoint^F_{\game,\objective}(\strategyone):=
\inf_{(\strategytwo,\disturbancestrategy) \in U_\strategyone} \expected_{\game,s_0}^{\strategyone,\strategytwo,\disturbancestrategy}\big(\longrundisturbances(\rho) \big).\]
When $U_\strategyone$ is empty, we also denote it as $\natural$.
\begin{definition}
	Given an SGD $\game$ and an objective $\phi$, the \emph{expected breaking point} of a Player~1 strategy $\pi$ is given by the pair 
	\[\expectedbreakpoint_{\game,\objective}(\strategyone):=\big(\expectedbreakpoint^T_{\game,\objective}(\strategyone), \expectedbreakpoint^F_{\game,\objective}(\strategyone)\big).\]
\end{definition}

Because of the observation in \cref{rem:disturbance_cost_equivalence} and \cref{lem:induced_mdp_equivalence}, we get the following lemma.
\begin{restatable}{lemma}{lemmaExpectedBPStrategy}
	\label{thm:expectedBPStrategy}
	Given an SGD, an objective $\phi$, and a memoryless Player~1 strategy $\strategyone$, it holds that
	$\expectedbreakpoint^T_{\game,\objective}(\strategyone) = \inf_{\mu\in U_\pi} \mathsf{TR}(M_\pi^\mu, s_0, C)$, and $\expectedbreakpoint^F_{\game,\objective}(\strategyone) = \inf_{\mu\in U_\pi} \mathsf{MP}(M_\strategyone^\mu,s_0, C)$.
\end{restatable}

\subsection{Computing the Expected Breaking Point}
\label{sec:expected_breaking_point_strategy}
Next, we describe the algorithm for expected transient and frequency breaking points for both safety and reachability objectives.
We compute the expected breaking point for memoryless strategies here\footnote{It can easily be extended to finite memory strategies using standard techniques, e.g., encoding the memory in state space by taking product of the game with a finite automaton representing the strategy.}.
We use the induced MDP $M_\strategyone$ and the cost function $C$ defined previously to compute it.

\paragraph{Safety}
To violate the safety objective, we need to compute the minimum expected cost while reaching states in $B$ with probability $\geq 1-p$.  
This is a variant of the standard SSP for which the solution can be found using the LP described in~\cite{SSP-tradeoff}.
If there is no solution to the linear program, it is not possible to reach $B$ with the required probability, which implies that the expected breaking point does not exist.
If a solution is found, the frequency breaking point is 0 and the transient breaking point is the solution to the LP.
Note that the strategy generated by the LP can be mixed.

\paragraph{Reachability}
\label{sec:expected_reachability}
To violate the reachability objective, we need to ensure the avoidance of states in the set $G$ with probability $\geq 1-p$ while minimizing cost. 
We proceed as follows.
Define $B$ to be the set of MECs in $M_\pi\setminus G$ where Player~2 can ensure to remain indefinitely without invoking any disturbances. 
These MECs can be identified by excluding any MEC of $M_\pi\setminus G$ that have exits via actions suggested by strategy $\strategyone$,
\begin{align}
	\label{eq:states_B}
	B:= &\{U\in \mathsf{MEC}(M_\strategyone \setminus G) \mid \forall s'\in U\cap \gamestatesone: \notag \\ 
	&\post[M_\strategyone](s',\strategyone(s')) \subseteq U\}.
\end{align}
We also define the set $R$ of MECs in which Player~2 can remain indefinitely but require disturbance actions to do so\footnote{We slightly abuse notation by using $B$ and $R$ to also refer to the union of these MECs.},
\begin{align}
	\label{eq:states_R}
	R:= &\{U\in \mathsf{MEC}(M_\strategyone\setminus G) \mid \forall s'\in U\cap \gamestatesone: \notag \\ 
	&\post[M_\strategyone](s',\strategyone(s')) \nsubseteq U \implies \gameavaildisturbanceactions(s)\neq \emptyset \}.
\end{align}
The probabilities of reaching $B$ and $R\cup B$ can be calculated through standard model checking algorithms.
Then, we identify 3 cases:

\textbf{Case 1: Probability to reach $B$ is $\geq 1-p$.}
The expected cost to reach $B$ is finite here, making the transient breaking point finite and the frequency breaking point 0.
Therefore, the transient breaking point is computed using an SSP formulation with respect to the reachability probability. 

\textbf{Case 2: Probability to reach $R\cup B$ is $< 1-p$.}
In this case, it is not possible to break the strategy as the goal states cannot be avoided with the required probability. 
Both the transient and frequency breaking points are $\natural$.

\textbf{Case 3: Probability to reach $B$ is $<1-p$ but $R\cup B$ is $\geq 1-p$.}
This is the intermediate case, where finitely many disturbances are not enough to break the strategy.
Intuitively, Player~2 cannot force the play to reach states (with the required probability) from which no more disturbances are needed to break the Player~1 strategy.
Therefore, the transient breaking point is $\omega$, and the expected frequency breaking point is computed as follows.
For each MEC in $R$, the minimum expected mean payoff of staying inside the MEC is computed. 
This mean payoff represents the frequency of disturbances required to stay within the MEC.
We then construct the weighted MEC quotient of $M_{\strategyone}$ (similar to \cite{ashok17meanpayoff}), by collapsing each MEC in $R$ to a single abstract state. 
A new cost function is defined for this quotient MDP.
Each collapsed state is augmented with an outgoing transition to a fresh terminal state $s_+$, where the new cost of this transition corresponds to the mean payoff of the original MEC. 
The new cost assigned to all other transitions in the MDP is zero.
Finally, we solve an SSP problem with the new cost on this modified MDP to compute the minimum expected cost of reaching $B\cup \{s_+\}$.
The solution to this SSP is precisely the expected frequency breaking point of $\strategyone$.

\begin{restatable}{theorem}{theoremExpectedBPStrategyReach}
	Given an SGD $\game$, objective $\phi$, Player~1 strategy $\strategyone$, and pair of values $(t,f)$, deciding whether $\expectedbreakpoint_{\game,\objective}(\strategyone)\geq (t,f)$ is in $\polytime$.
\end{restatable}

\subsection{Optimal Strategy: Expected Breaking Point}
\label{sec:expected-synthesis}

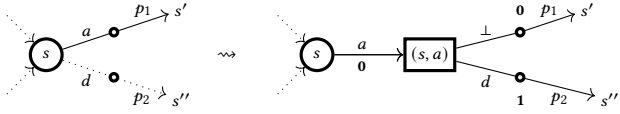
\begin{figure}[t]
	\scriptsize
	\centering

	\begin{tikzpicture}[scale=0.6]
		
		\node[state] at (0,0)     (0)   {$s$};
		\node[] at (-1,1) (1) {};
		\node[] at (-1,-1) (2) {};
		\node[] at (3,1) (3) {$s'$};
		\node[] at (3,-1) (4) {$s''$};
		\node[smallcircle] at (1.5,-0.5) (circle1) {};
		\node[smallcircle] at (1.5,0.5) (circle2) {};

		\node[state] at (6,0)     (5)   {$s$};
		\node[] at (5,1) (6) {};
		\node[] at (5,-1) (7) {};
		\node[squarednode] at (8.5,0)     (9)   {$(s, a)$};
		\node[] at (12.5,-1) (10) {$s''$};
		\node[] at (12,1) (11) {$s'$};
		\node[smallcircle] at (10.5,-0.5) (circle3) {};
		\node[smallcircle] at (10.5,0.5) (circle4) {};

		\node[] at (4,0) {$\rightsquigarrow$};

		\node at (10.5,1) {\textbf{0}};
		\node at (10.5,-1) {\textbf{1}};
		
		\draw
		(1) edge[->, dotted] node {} (0)
		(2) edge[->, dotted] node {} (0)
		(0) edge[-, above] node {$a$} (circle2)
		(0) edge[-, dotted, below] node {$d$} (circle1)
		(circle2) edge[->, above] node {$p_1$} (3)
		(circle1) edge[->, dotted, below] node {$p_2$} (4)
		
		(6) edge[->, dotted] node {} (5)
		(7) edge[->, dotted] node {} (5)
		(9) edge[-, below] node {$d$} (circle3)
		(circle3) edge[->, below] node {$p_2$} (10)
		(9) edge[-, above] node {$\bot$} (circle4)
		(circle4) edge[->, above] node {$p_1$} (11)
		(5) edge[->, above, pos=0.4] node {$a$} (9)

		(5) edge[->, below, pos=0.4] node {\textbf{0}} (9)
		;
	\end{tikzpicture}

	\caption{Gadget for converting an SGD into an SG for the Expected Breaking Point.
	}
	\label{fig:gadget_expected_breaking_point}
\end{figure}

The optimally resilient strategy is defined as the strategy that achieves the breaking point 

\[
\expectedbreakpoint_{\game,\objective} = \big(\max_\strategyone \expectedbreakpoint^T_{\game,\objective}(\strategyone), \max_\strategyone \expectedbreakpoint^F_{\game,\objective}(\strategyone)\big).
\]
To compute this, we construct a transformed stochastic game $\newgame$ from the given SGD $\game$ using the gadget in \cref{fig:gadget_expected_breaking_point}. 
This transformation follows the unfolded SG structure introduced in \cref{sec:unfolded-SG}, but omits the use of duplicated copies.
Specifically, for each state $s \in \gamestatesone$ and action $a \in \gameavailactions(s)$, we introduce an intermediate Player~2 state $(s, a)$, reached deterministically when Player~1 selects action $a$ in state $s$. 
From $(s, a)$, Player~2 chooses either $\bot$, that follows the distribution $T(s, a)$ simulating the undisturbed execution of $a$, or a disturbance action $d$, which follows the distribution $T^D(s, d)$, simulating the effect of a disturbance. 
To capture the cost, disturbance actions are assigned cost 1, and all others cost 0, as indicated by the boldface labels in \cref{fig:gadget_expected_breaking_point}. 
The minimum cost incurred while violating $\phi$ in $\newgame$ corresponds to the breaking point.
\begin{new}
Thus, optimally resilient Player~1 strategies are generally memoryless and randomized for the expected breaking point.
Optimal disturbance strategies are also memoryless and randomized; for instance, see \Cref{fig:worst_vs_average_breaking_point}.
\end{new}

\begin{restatable}{lemma}{theoremExpectedBPSynthesis}
	\label{thm:expectedBPSynthesis}
	Given an SGD $\game$, objective $\phi$ and memoryless Player~1 strategy $\strategyone$, it holds that
	$\expectedbreakpoint^T_{\game,\objective}(\pi) = \inf_{\mu\in U_\pi} \mathsf{TR}(\newgame^{\strategyone, \mu}, s_0, C)$ and $\expectedbreakpoint^F_{\game,\objective}(\pi) = \inf_{\mu\in U_\pi} \mathsf{MP}(\newgame^{\strategyone, \mu},s_0, C)$.
\end{restatable}

The algorithm's structure parallels the previous algorithm, with the main difference being that it solves an SG rather than an MDP.

\paragraph{Safety}
For safety objectives, first the maximum probability to reach $B$ is computed via a QP with Player~2 as the maximizer. 
If said probability is $\geq 1-p$, the problem reduces to solving an SSP problem in the transformed game $\newgame$, where the goal is to reach the target set $B$ with $\geq 1-p$ probability.
The SSP problem can be solved using VI or PI
\footnote{The current algorithms for SSP using VI and PI only give approximate solutions with convergence guarantees in the limit. We do not get the exact breaking point in this case, however, if exact solvers for SSP are found in the future, our approach can use them off-the-shelf.}
~\cite{patek1999stochastic}.
The SSP problem usually assumes that the target sink (here it's $B$) is reached with probability 1 to ensure that the expected cost is finite. 
For our case, since there are zero cost paths always available, the same algorithm can be used.
Here as soon as the play reaches $G$, the play can be stopped, and the cost is 0.
The solution yields the transient breaking point and the frequency breaking point is 0.
If the probability to reach $B$ is $< 1-p$, it indicates that the strategy is not breakable.

\paragraph{Reachability}
For reachability objectives, the algorithm first identifies the set of states in $\newgame$ from which Player~2 can ensure that the target set $G$ is avoided with probability 1. 
This is done via a backward fixed-point computation, starting with the set $E = \gamestates \setminus G$ and iteratively removing states according to the following rules until convergence:
\begin{itemize}
	\item A Player~1 state is removed if there exists an action leading outside $E$ with positive probability.
	\item A Player~2 state is removed if all available actions lead outside $E$ with positive probability.
\end{itemize}

Next, we compute the maximum probability of reaching $E$ in $\newgame$ using QP.
If the reachability probability is $< 1 - p$, the strategy cannot be broken, and the algorithm returns $(\natural, \natural)$.
If the reachability probability is $\geq 1 - p$, we proceed to analyze the MECs contained in $E$.
For each MEC, we compute its minimum mean payoff, which corresponds to the minimal average disturbance cost required to remain within that component indefinitely~\cite{sg-mean-payoff}. 
Let $B \subseteq E$ denote the set of MECs with zero mean payoff.

If $B$ is reachable with probability $\geq 1 - p$, the transient breaking point is finite and the frequency breaking point is $0$, indicating that the strategy can be broken with finitely many disturbances.
This is computed by solving an SSP problem in $\newgame$ with the target set $B$ and probability threshold $1 - p$.
In the other case, the transient breaking point is $\omega$, and the frequency breaking point requires the use of weighted quotient stochastic game that is constructed by collapsing each MEC into a single state. 
A new terminal state $s_+$ is introduced and an outgoing transition from all collapsed states to $s_+$ is added, with the cost of this transition equal to the mean payoff of the corresponding MEC.
The cost of all other transitions is set to zero.
Finally, we solve an SSP problem in this quotient SG to compute the frequency breaking point.

\begin{restatable}{theorem}{theoremExpectedBPSReach}
	Given an SGD $\game$, objective $\phi$, and pair of values $(t,f)$, deciding whether $\expectedbreakpoint_{\game,\objective}\geq (t,f)$ is as hard as solving SSP for SGs.
\end{restatable}

\section{Worst-case Breaking Point}
\label{sec:worst_case_breaking_point}
In contrast to the expected case, we now consider the worst-case number of disturbances required to break a strategy over all possible paths.

Now, the \emph{worst-case transient breaking point} of a strategy is 
\[	
	\breakpoint^T_{\game,\objective}(\strategyone):=
	\inf_{(\strategytwo,\disturbancestrategy) \in U_\strategyone} \inf \Big\{x\in \langle \omega\rangle ~\big|~ P_{\game,s_0}^{\strategyone,\strategytwo,\disturbancestrategy}\big(\rho \mid \disturbances(\rho) \leq x \big) = 1\Big\}.
\]
We only require that almost all paths (i.e., with probability 1) have fewer than $k$ disturbances, reflecting the common probabilistic convention of ignoring measure-zero events that do not affect the typical behaviors.
Recall again that when the set $U_\strategyone$ is empty, we define $\breakpoint^T_{\game,\objective}(\strategyone)$ to be $\natural$, and when there exist no finite values in this set, we denote $\breakpoint^T_{\game,\objective}(\strategyone)$ to be $\omega$.
When the transient breaking point is $\omega$, we can compute the frequency of disturbances required.
The \emph{worst-case frequency breaking point} of a Player~1 strategy $\strategyone$ is 
\[
	\breakpoint^F_{\game,\objective}(\strategyone):=
	\inf_{(\strategytwo,\disturbancestrategy) \in U_\strategyone} \inf \Big\{x\in[0,1] ~\big|~ P_{\game,s_0}^{\strategyone,\strategytwo,\disturbancestrategy}\big(\rho \mid \longrundisturbances(\rho) \leq x \big) = 1\Big\}.
\]

We can now combine the two values and define the \emph{worst-case breaking point} of a strategy.
\begin{definition}
	Given a SGD $\game$ and an objective $\phi$, the \emph{worst-case breaking point} of a Player~1 strategy $\strategyone$ is given by the pair
	$\breakpoint_{\game,\objective}(\strategyone):=
	\Big(\breakpoint^T_{\game,\objective}(\strategyone), \breakpoint^F_{\game,\objective}(\strategyone)\Big)
	$.
\end{definition}

\subsection{Computing the Worst-Case Breaking Point} 
\label{sec:worst_case_breaking_point_strategy}
We describe the algorithm for transient and frequency worst-case breaking points for both safety and reachability objectives.
We compute the breaking points for a memoryless strategy $\strategyone$ here, but it can easily be extended to finite-memory strategies by encoding the memory in the state space, similar to the expected case.  
First, we provide an overview of the key steps, deferring the detailed algorithms for computing the transient and frequency breaking points to later in this section. 
If $\phi$ is a safety objective, set $B$ is given as input. 
And, if $\phi$ is a reachability objective, we compute the set $B$ as defined in \cref{eq:states_B}.
We now compute the maximum probability in the induced MDP $M_\pi$ of reaching $B$ using LP.
Based on the result, we distinguish three cases: the probability is either $>1-p$, $<1-p$, or $=1-p$.

\textbf{Case 1: Probability to reach $B$ is $>1-p$.}
In this case, the strategy is breakable with finitely many disturbances. 

\textbf{Case 2: Probability to reach $B$ is $=1-p$.}
In this case, if the strategy is breakable in a finite number of disturbances, this number should be bounded by $k$: the number of disturbance actions present in the graph.
We run the procedure to compute the worst-case transient breaking point for $|T^D|$ iterations.
Based on the outcome of this procedure, it is further divided into two subcases:
\begin{itemize}
	\item[a.] If the procedure terminates at the $i^{th}$ iteration for some $i$, we have $\breakpoint^T_{\game,\objective}(\strategyone)=i$. 
	\item[b.] If it does not terminate within $|T^D|$ iterations,
	it requires infinitely many disturbances, making the transient breaking point $\omega$ and the frequency breaking point $0$.
	This is due to the fact that there are probabilistic loops that are required to reach $B$ with a $1-p$ probability, whereas the frequency of taking that action in the long term would still be $0$.
\end{itemize}

\textbf{Case 3: Probability to reach $B$ is $<1-p$.}
Here, if $\phi$ is a safety objective, the strategy is not breakable with any amount of disturbances.
If $\phi$ is a reachability objective, we compute the set $R$ using \cref{eq:states_R}.
If the probability of reaching $R \cup B$ is $<1-p$, again, the strategy is not breakable with any amount of disturbances and the algorithm returns $(\natural, \natural)$.
If the probability of reaching $R \cup B$ is $\geq 1-p$, the strategy is breakable but requires infinitely many disturbances.
This results in the transient breaking point being $\omega$ and the frequency breaking point being computed using the procedure described later.

\subsubsection{Transient Breaking Point}
\label{sec:computing_transient_worst_case}
Given $k$,\footnote{\label{footnote:upper_bound}An upper bound here, if it is finite, can be computed by counting the number of iterations required by VI for MDPs to go beyond probability $1-p$ of reaching $B$.} we want to verify if $\breakpoint^T_{\game,\objective}(\strategyone)\leq k$.
We define a sequence of $k$ reachability LPs where the $i^{th}$ LP checks if the strategy $\strategyone$ can be broken using at most $i$ disturbances. 
The $i^{th}$ LP uses the solution of the $(i-1)^{st}$ LP and allows one more disturbance to compute the probability of reaching $B$ with at most $i$ disturbances.
If the solution of $i^{th}$ LP is $\geq 1 - p$, $\strategyone$ is breakable using $i$ disturbances.
For completeness, the explicit LP can be found in the supplementary material.

Each LP is solved in polynomial time, and each iteration only requires the result of the previous LP. 
This gives us an algorithm that is polynomial in terms of $|\game|$ and $k$. 
This gives us a parametrized polytime complexity where $k$ is given in unary, and in general a $\PSPACE$ algorithm.
It can also be shown that the optimal disturbance strategy here may require memory of size $k$.

\subsubsection{Frequency Breaking Point}
\label{sec:frequency_worst_case}
Recall that the frequency breaking point is only computed when the specification is reachability and the probability of reaching $B$ is $<1-p$ but the probability of reaching $R\cup B$ is $\geq 1-p$. 

As in the expected case, the procedure assigns a frequency of disturbances required to stay within the set $B$ as 0 and assigns the mean payoff for the MECs in the set $R$ as their frequency.
In contrast to the expected case, where we solved a SSP to find the final expected value, we need to consider the worst case.
Here, we iteratively remove a MEC with the highest disturbance frequency. 
After each removal, it recomputes the probability of reaching the remaining MECs in $R\cup B$. 
This process continues until the probability of reaching the remaining MECs drops below $1 - p$.
The frequency of disturbances required to remain in the last removed MEC is then returned as the frequency breaking point.

\begin{restatable}{theorem}{theoremBPWC}
	Given an SGD $\game$, objective $\phi$, Player~1 strategy $\strategyone$, and pair of values $(t,f)$, deciding if $\breakpoint_{\game,\objective}(\strategyone)\geq (t,f)$ is in $\PSPACE$.
\end{restatable}

\subsection{Optimal Strategy: Worst-Case Breaking Point}
\label{sec:worst_case_synthesis}
An optimally resilient strategy is a strategy that achieves the following worst-case breaking point:
\[
\breakpoint_{\game,\objective} = \big(\max_\strategyone \breakpoint^T_{\game,\objective}(\strategyone), \max_\strategyone \breakpoint^F_{\game,\objective}(\strategyone)\big)
\]
\begin{added}
We first discuss memory requirements for transient and frequency breaking points, and later provide the algorithm to compute them.

\subsubsection{Memory Requirements for Transient Breaking Point}\label{sec:memory_worst_case}
Making use of the reduction to standard stochastic games from \Cref{sec:unfolded-SG}, we next show that deciding the worst-case transient breaking point is equivalent to solving the unfolded stochastic game.
This allows us to use results concerning stationarity in standard SGs to reason about memory requirements in SGDs.

\begin{figure}[t]
	\scriptsize
	\centering
	\begin{minipage}{0.23\textwidth}
		\centering
		\begin{tikzpicture}[scale=0.7]
			% Node styles
			\node[state, initial, initial text=] at (-2,2) (s0) {$s_0$};
			\node[state, initial, initial text=] at (-2,0) (s1) {$s_1$};
			\node[state] at (-0.5,1) (s2) {$s_2$};
			\node[state] at (-0.5,-1) (s3) {$s_3$};
			\node[state, fill=our-green] at (2,1) (goal) {G};
			\node[state, fill=our-red] at (2,-1) (bad) {B};
			\node[smallcircle] at (0.75,1) (circle1) {};
			
			% Arrows with labels
			\draw[->, left, dashed] (s0) edge node {$d$} (s1);
			\draw[->, bend left, above] (s0) edge node {$a$} (goal);
			\draw[->, above] (s1) edge node {$a_1$} (s2);
			\draw[->, above] (s1) edge node {$a_2$} (s3);
			\draw[-, above] (s2) edge node {$a$} (circle1);
			\draw[->, above] (circle1) edge node {$0.5$} (goal);
			\draw[->, right] (circle1) edge node {$0.5$} (bad);
			\draw[->, above] (s3) edge node {$a$} (goal);
			\draw[->, above, dashed] (s3) edge node {$d$} (bad);
			
			% Initial distribution
			\node at (-2.75,2.3) {0.5};
			\node at (-2.75,0.3) {0.5};
		\end{tikzpicture}
	\end{minipage}
	\hfill
	\begin{minipage}{0.23\textwidth}
		\centering
		\begin{tikzpicture}[scale=0.7]
			\node[state, initial, initial text=] at (0,0)     (1)   {$s_1$};
			\node[state] at (1.5,0)     (2)   {$s_2$};
			\node[state] at (3,0)     (3)   {$s_3$};
			\node[state, fill=our-green] at (1.5,-1.8)     (G)   {G};
			\node[state, fill=our-red] at (4.5,-1.8)    (B)   {B};
			\node[smallcircle] at (4.5,0)    (c1)   {};
			\node[smallcircle] at (0,-1.8)    (c2)   {};
			
			\draw
			(1) edge[->, above] node {$a$} (2)
			(1) edge[-, right, dashed] node {$d$} (c2)
			(2) edge[->, right] node {$a$} (G)
			(2) edge[->, above, dashed] node {$d$} (3)
			(3) edge[->, right] node {$a$} (G)
			(3) edge[-, above, dashed] node {$d$} (c1)
			
			(c1) edge[->, right, dashed] node {$0.5$} (B)
			(c1) edge[->, bend right, above, dashed] node {$0.5$} (1)
			
			(c2) edge[->, above, dashed] node {$0.5$} (G)
			(c2) edge[->, bend right, below, dashed] node {$0.5$} (B);
		\end{tikzpicture}
	\end{minipage}
	\caption{An SGD where the most resilient $\strategyone$ must have memory even if $\delta$ is memoryless (left), and
	one where the optimal $\delta$ must rely on memory even if $\strategyone$ is memoryless (right).}
	\label{fig:memory_counterexamples}
\end{figure}
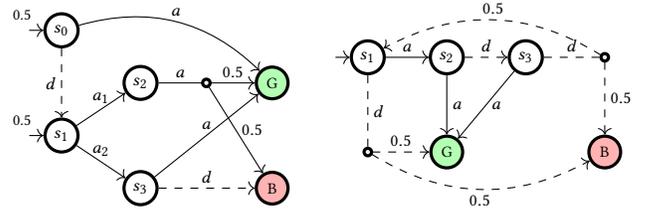

\begin{restatable}{lemma}{worstCaseTransBP}
	\label{thm:equivalence_transient}
	Let $\game$ be an SGD and $\game^{\dagger k}$ its corresponding \mbox{$k$-unfolded} stochastic game, for some $k \in \langle\omega\rangle$.
	Then, for all Player~1 strategies $\pi$, we have $\breakpoint^T_{\game, \phi}(\pi) \leq k$ if and only if $\sup_{\tilde \sigma} \prob^{\tilde \pi, \tilde{\sigma}}_{\game^{\dagger k}, (s_0,k)}(\neg \phi) \geq 1 - p$.
\end{restatable}

If we let $k$ be the optimal breaking point, $k = \max_{\pi \in \Pi} \breakpoint^T_{\game, \phi}(\pi)$, then the unfolded game $\game^{\dagger k}$ is a standard stochastic game, and the optimal strategies for both players are therefore memoryless~\cite{shapley1953stochastic}. 
\Cref{thm:equivalence_transient}, together with \Cref{thm:unfolded_sg_equivalence}, implies that an SGD $\game$ is stationary with respect to the state-counter pair $(s, i)$. 
Here, $i$ represents the number of disturbances remaining, and the pair $(s, i)$ uniquely identifies a state in the corresponding unfolded game. 
Consequently, we can derive the following two corollaries.

\begin{corollary}\label{cor:most_resilient}
    For any SGD $\game$, there exists an optimally resilient Player~1 strategy of the form $\pi^* : S \times \langle k+1 \rangle \to A$. 
\end{corollary}

\begin{corollary}\label{cor:optimal_delta}
    For any SGD $\game$, there exists an optimal $k$-disturbance strategy of the form $\delta^* : S \times \langle k+1 \rangle \to A$. 
\end{corollary}

In the next two examples, we construct SGDs where the optimal strategies require memory.

\begin{example}\label{ex:memory_counterex}
Let $\game$ be as in \Cref{fig:memory_counterexamples} (left) with $\phi = P_{\geq 0.75}(\lozenge\text{G})$.
There is a step-counting strategy of the form $\pi^* : S \times \langle 2 \rangle \to A$ with a worst-case transient breaking point of 2:
$\pi^*(s_0, \_) = a$, $\pi^*(s_1, 0) = d$, $\pi^*(s_1, 1) = a$, $\pi^*(s_3, \_) = a$, and $\pi^*(s_2, \_) = a$.
In contrast, no memoryless Player~1 strategy can achieve a breaking point of 2.
\end{example}

\begin{example}\label{ex:disturb_not_early}
Let $\game$ be as in \Cref{fig:memory_counterexamples} (right) with $\phi = P_{\geq 0.5}(\lozenge\text{G})$.
Define the memoryless Player~1 strategy $\pi(s) = a$ for every state $s$. 
Then, $\pi$ has the worst-case transient breaking point of 3, as the following step-counting 3-disturbance strategy breaks it: 
$\delta(s_1, 3) = \bot, \delta(s_1, 2) = \delta(s_1, 1) = d, \delta(s_2, \_) = d, \delta(\_, 0) = \bot$.
Yet, no memoryless disturbance strategy that disturbs at most 3 times can break $\pi$.
\end{example}

In general, it is \emph{necessary} for optimal strategies to perform step counting (i.e., their memory is of the form $S \times \langle k + 1 \rangle$), as seen in \Cref{ex:memory_counterex,ex:disturb_not_early}, while \Cref{cor:most_resilient,cor:optimal_delta} show this is \emph{sufficient}.
Thus, we always assume that optimal strategies have this form in the case of worst-case transient resilience without loss of generality.

\subsubsection{Memory Requirements for Frequency Breaking Point}
\label{app:memoryless_strategies_worst_case_frequency}
In contrast to the transient case, memoryless strategies are sufficient for the frequency breaking point. 
We discuss this for reachability objectives; in the case of safety objectives either it can be broken with finite disturbances or it cannot be broken at all.
\begin{restatable}{lemma}{worstCaseFreqBP}
\label{lem:worstCaseFreqBP}
	If $\phi$ is reachability, and it is not possible to break the strategy with finite disturbances, then the optimal disturbance strategy is memoryless.
\end{restatable}

The intuition behind \Cref{lem:worstCaseFreqBP} is that, since the strategy cannot be broken with finite disturbances, the disturbance strategy must rely on staying in some MECs of the MDP induced by the Player~1 strategy.
In these MECs, the optimal disturbance strategy to minimize the mean payoff is memoryless.
Thus, the overall disturbance strategy can be chosen to be memoryless.
\end{added}

\subsubsection{Computing Optimally Resilient Strategies}

The procedure here is similar to the procedure in \cref{sec:expected-synthesis}.
Here also, we use the transformation to $\newgame$.
If $\phi$ is safety, then the probability of reaching $B$ is computed. 
If it is $<1-p$, then the strategy is not breakable.
In the other case, the frequency breaking point is 0 and the transient breaking point is computed as described later in this section.

If $\phi$ is reachability, then the algorithm computes the set of states $E$ from which Player~2 can force the play to never reach $G$. 

This can be computed using the fix-point computation as described in \cref{sec:expected-synthesis}.
If the maximum probability of reaching $E$ is $<1-p$, then the strategy is again not breakable.
If it is $\geq 1-p$, the minimum mean payoff of each MEC in $E$ is computed, which gives us the minimum frequency of disturbances required to stay in that MEC.

Let $B$ be the set of MECs in $E$ with 0 mean payoff and if $B$ is reachable with probability $\geq 1-p$, then the frequency breaking point is 0 and the transient breaking point is computed as described next.
Otherwise, the algorithm iteratively removes a MEC with the highest disturbance frequency.
After each removal, it recomputes the probability of reaching the remaining MECs in $E$.
This process continues until the probability of reaching the remaining MECs drops below $1 - p$.
The frequency required to remain in the last removed MEC is then returned as the frequency breaking point. 

\subsubsection{Transient Breaking Point}
\label{sec:qp}

Given $k$, we want to verify whether $\breakpoint^T_{\game,\objective}\leq k$.
We define a sequence of $k$ QPs where the $i^{th}$ QP checks if all strategies $\strategyone$ can be broken using at most $i$ disturbances.
Here too, the $i^{th}$ QP uses the solution of the $(i-1)^{th}$ QP, and by allowing one more disturbance, it computes the probability of reaching $B$ with $i$ disturbances. 
\begin{new}
	Intuitively, the $i^{th}$ QP solves for the maximum probability of reaching bad states using the method in \citep{condon1992complexity} in the game unfolded for $i$ disturbances (\Cref{sec:unfolded-SG}) but reuses results from the game unfolded for $i-1$ disturbances.
\end{new}
Description of these QPs is provided in the Supplementary Material.

Note that the decision problem corresponding to these specific QPs can be solved in $\NP$~\cite{condon1992complexity}, but in our case we also need to find the solution of the QP. 
This can be extracted using polynomially many $\NP$ queries, where in $j^{th}$ query we will ask whether the $j^{th}$ bit of the solution is $1$. Also, in each iteration, we only need to remember the solution of the last iteration. 
Thus, we can solve this in polynomial time with an $\NP$ oracle if $k$ is given in unary.
\begin{new}
	In practice, the $\NP$-oracle complexity would involve querying a SAT solver polynomially many times to extract the solution bit-by-bit, and we would expect the state-of-the-art SAT-solvers to scale well. 
\end{new}

\begin{restatable}{theorem}{theoremFreqResNP}
	\label{thm:freq_res_np} 
	Given an SGD $\game$, an objective $\phi$, and a pair of values $(t,f)$, deciding whether $\breakpoint_{\game,\objective}\geq (t,f)$ is in $\PSPACE$.
\end{restatable}

\section{Conclusions and Future Work}
In this work, we explored the concept of resilient strategies in stochastic systems to analyze their robustness against disturbances. 
We introduced novel formulations for resilience by considering both expected and worst-case breaking points, and refined the notion of resilience through measures such as frequency-based disturbances.
We provided algorithms for computing resilience, offering solutions for both expected and worst-case scenarios, and
highlighted the trade-offs between optimality and resilience.
\begin{new}
By introducing this mathematical framework, as well as providing proofs of properties of general theoretical interest, our work can serve as a foundation for developing practical solutions to real-world problems.
\end{new}
Future work includes extending these concepts to partially observable Markov decision processes (POMDPs), multi-agent settings, and alternative objectives beyond reachability and safety.

%%%%%%%%%%%%%%%%%%%%%%%%%%%%%%%%%%%%%%%%%%%%%%%%%%%%%%%%%%%%%%%%%%%%%%%%

%%% The acknowledgments section is defined using the "acks" environment
%%% (rather than an unnumbered section). The use of this environment 
%%% ensures the proper identification of the section in the article 
%%% metadata as well as the consistent spelling of the heading.

\begin{acks}
% If you wish to include any acknowledgments in your paper (e.g., to 
% people or funding agencies), please do so using the `\texttt{acks}' 
% environment. Note that the text of your acknowledgments will be omitted
% if you compile your document with the `\texttt{anonymous}' option.

Kush Grover was partially supported by the STEP-RL project funded by the European Research Council under GA number 101115870.
Markel Zubia and Nils Jansen were partially supported by the European Research Council Starting Grant 101077178 (DEUCE).
Debraj Chakraborty was partially supported by the National Research Foundation, Singapore, under its RSS Scheme (NRF-RSS2022-009), and the Intelligence-Oriented Verification \& Controller Synthesis (InOVationCS) project funded by the European Union under Grant Agreement number 101171844.
Jan K\v{r}etinsk'y was partially supported by the DFG project GOPro funded by the German Research Foundation (DFG) under project number 427755713, the MUNI Award in Science and Humanities (MUNI/I/1757/2021), and the Intelligence-Oriented Verification \& Controller Synthesis (InOVationCS) project funded by the European Union under Grant Agreement number 101171844.
\end{acks}

%%%%%%%%%%%%%%%%%%%%%%%%%%%%%%%%%%%%%%%%%%%%%%%%%%%%%%%%%%%%%%%%%%%%%%%%

%%% The next two lines define, first, the bibliography style to be 
%%% applied, and, second, the bibliography file to be used.
\balance
\bibliographystyle{ACM-Reference-Format} 
\bibliography{references.bib}

\clearpage
\appendix

\section{Stochastic Games with Disturbances (Section~\ref{sec:sgd})}
\label{app:induced-MDP}

\subsection{Induced MDP under a Player 1 Strategy}
\label{app:induced-MDP-transformation}
\begin{definition}
	
For a Player 1 memoryless strategy $\strategyone$ in an SGD $\game$, the induced MDP is defined as $M_\strategyone = (\bar{S}, \bar{A}, \bar{Av}, \bar{T})$, where 
\begin{itemize}
	\item $\bar{S} = \gamestatesone \cup \gamestatestwo$
	\item $\bar{A} = \gameactions \cup \gamedisturbanceactions$
	\item $\bar{Av}(s) =
			\begin{cases}
				\{\strategyone(s)\} \cup \gameavaildisturbanceactions(s) & s\in \gamestatesone \\
				\gameavailactions(s) & s\in \gamestatestwo
			\end{cases} $
	\item $\bar{T}(s,a,s') =
		\begin{cases}
			\gametransitions(s,a,s') & s\in \gamestatesone, a=\strategyone(s) \\
			\gamedisturbancetransitions(s,a,s') & s\in \gamestatesone, a\in \gameavaildisturbanceactions(s) \\
			\gametransitions(s,a,s') & s\in \gamestatestwo, a\in \gameavailactions(s)
		\end{cases}$
\end{itemize} 
\end{definition}

\label{app:induced-MDP-proof}
\theoremInducedMDPEquivalence*
\begin{proof}
	
	\textbf{Forward direction:}
	Given induced MDP $M_\strategyone$ and strategies $\strategytwo, \disturbancestrategy$, 
	$h$ maps them to a strategy $\mu$ of $M_\strategyone$ which is defined as follows: 
	\[\mu(w\cdot s) = 
		\begin{cases}
			\strategyone(w\cdot s) & s\in\gamestatesone, \disturbancestrategy(w\cdot s) = \bot \\
			\disturbancestrategy(w\cdot s) & s\in\gamestatesone, \disturbancestrategy(w\cdot s) \neq \bot \\
			\strategytwo(w\cdot s) & s\in \gamestatestwo
		\end{cases}\]
	Therefore, $h(\strategytwo, \disturbancestrategy) = \mu$.

	Note that for each finite path of $\game^{\strategyone, \strategytwo, \disturbancestrategy}_{s_0}$ there is an equivalent path in $M^{\mu}_{\pi, s_0}$ that has the same probability and reaches the same state.
	This gives us that for all paths in $\game^{\strategyone, \strategytwo, \disturbancestrategy}_{s_0}$ that reaches $B$, there are equivalent paths in $M^{\mu}_{\pi, s_0}$ that reaches $B$.
	This also shows that the probability of reaching $B$ is the same in both.

	\noindent
	\textbf{Backward direction:}
	Given a strategy $\mu$ for $M_\strategyone$, we construct a Player~2 strategy and a disturbance strategy as follows: 
	\[\strategytwo(w\cdot s) = \mu(w\cdot s)\quad ~~s\in \gamestatestwo\]
	\[\disturbancestrategy(w\cdot s)=
		\begin{cases}
			\bot & s\in\gamestatesone, \mu(w\cdot s) \in \gameavailactions(s)\\
			\mu(w\cdot s) & s\in\gamestatesone, \mu(w\cdot s) \in \gameavaildisturbanceactions(s) \\
			
		\end{cases}\]
	Again, for each path of $M_{\pi, s_0}$ there is an equivalent path in $\game^{\strategyone, \strategytwo, \disturbancestrategy}_{s_0}$ that reaches the same state with same probability.
	With the same argument as the previous case, the probability of reaching $B$ is same in both.
	Note that in this case, $h(\strategytwo, \disturbancestrategy) = \mu$, proving that $h$ is a bijection.
\end{proof}

This immediately gives us the following corollary.
\begin{corollary}
	\label{lem:finite_breakpoint_mdp_equivalence}
	Given a set of states $B\subseteq\gamestates$. 
	For each set of game strategies $(\strategyone, \strategytwo, \disturbancestrategy)$, there exist a strategy $\mu$ for $M_\strategyone$ such that $P_{\game, s_0}^{\strategyone, \strategytwo, \disturbancestrategy}(\lozenge B)$ = $P_{M_\strategyone, s_0}^{\mu}(\lozenge B)$ and vice versa.
\end{corollary}

\subsection{Converting SGDs with Finite Disturbances to SGs}

\label{app:unfolded-SG}
\begin{definition}\label{def:unfolded-SG}
	Let $\game = (\gamestates, \gamestatesone, \gamestatestwo,$ $\gameactions, \gameavailactions, \gametransitions, \gamedisturbanceactions, \gameavaildisturbanceactions, \gamedisturbancetransitions)$ be an SGD and let $k \in \langle \omega\rangle$.
	Its corresponding $k$-unfolded stochastic game is an SG $\game^{\dagger k} = (\gamestates^{\dagger k}, \gamestatesone^{\dagger k}, \gamestatestwo^{\dagger k}, \gameactions^{\dagger k}, \gameavailactions^{\dagger k}, \gametransitions^{\dagger k})$ where
	\begin{itemize}
		\item $\gamestatesone^{\dagger k} = \gamestatesone \times \langle k + 1 \rangle$.
		\item $\gamestatestwo^{\dagger k} = \gamestatestwo \times \langle k + 1 \rangle \cup \gamestatesone \times \langle k + 1 \rangle \times A$.
		\item $\gamestates^{\dagger k} = \gamestatesone^{\dagger k} \cup \gamestatestwo^{\dagger k}$.
		\item $\gameactions^{\dagger k} = \gameactions \cup \gamedisturbanceactions \cup \{\bot\}$.
		\item \[\gameavailactions^{\dagger k}(s) = 
		\begin{cases}
			\gameavailactions(\tilde s) &\text{ if } s = (\tilde s, i), \tilde s \in \gamestatesone\\
			\gameavailactions(\tilde s) &\text{ if } s = (\tilde s, i), \tilde s \in \gamestatestwo\\
			\gameavaildisturbanceactions(\tilde s) \cup \{\bot\} &\text{ if } s = (\tilde s, i, a), \tilde s \in \gamestatesone\\
		\end{cases}	
		\]
		\item $\gametransitions^{\dagger k}(s, a, s') =  \\
		\begin{cases}
			\gametransitions(\tilde s, a, \tilde s') & s = (\tilde s, i), \tilde s \in \gamestatestwo, s' = (\tilde s', i) \\
			0 &s = (\tilde s, i), \tilde s \in \gamestatestwo \\
			\gametransitions(\tilde s, a, \tilde s') &s = (\tilde s, 0), \tilde s \in \gamestatesone, s' = (\tilde s', 0) \\
			0 &s = (\tilde s, 0), \tilde s \in \gamestatesone \\
			1 &s = (\tilde s, i+1), \tilde s \in \gamestatesone, s' = (\tilde s, i+1, a) \\
			0 &s = (\tilde s, i+1), \tilde s \in \gamestatesone\\
			\gametransitions(\tilde s, \tilde a, \tilde s') &s = (\tilde s, i, \tilde a), a = \bot, s' = (\tilde s', i) \\
			0 &s = (\tilde s, i, \tilde a), a = \bot \\
			\gamedisturbancetransitions(\tilde s, \tilde a, \tilde s') &s = (\tilde s, i+1, \tilde a), \tilde a \in \gamedisturbanceactions, s' = (\tilde s', i)\\
			0 &s = (\tilde s, i, \tilde a), \tilde a \in \gamedisturbanceactions
		\end{cases}
		$
	\end{itemize}
\end{definition}

\theoremUnfoldedSGEquivalence*
\begin{proof}
	Notice that the difference between $\game$ and $\game^{\dagger k}$ is the use of the gadget shown in \Cref{fig:gadget_unfold} that converts disturbance edges into Player~2 states. 
	Given a path $\rho$ in $\game$, we can get $g(\rho) = \tilde \rho$ by doing the following:
	\begin{itemize}
		\item Each Player~2 state $s$ in $\rho$ is replaced with a tuple $(s, i)$, where $i$ is the number of disturbance edges $a \in \gamedisturbanceactions$ present in $\rho$ up until $s$, 
		\item Each Player~1 state $s$ in $\rho$ is replaced with a tuple $(s, i)$, similarly to the previous step. 
		If $s$ is followed up in $\rho$ by a player action $a \in A$, we can replace it by the action $a$, a visit to the node $(s, i, a)$, and action $\bot$. 
		Otherwise, if $s$ is followed up by a disturbance action $d \in \gamedisturbanceactions$, we must replace it by the action $\pi(a)$, a visit to node $(s, i, a)$, action $d$, and decrement the counter of the successor state by 1.
		Since $\delta$ disturbs at most $k$ times by assumption, the counter can never go below 0.
	\end{itemize}
	
	We obtain $\tilde \pi$ by just letting $\tilde \pi(\tilde \rho) = \pi(\rho)$, while $\tilde \sigma(\tilde \rho)$ is just a combination of $\sigma$ and $\delta$, depending on whether we are on a $(s, i)$ or a $(s, i, a)$ node. 
	By construction of $\game^{\dagger k}$, it is clear that $\rho$ and $\tilde \rho$ have the same probabilities under the respective policies.
\end{proof}

\section{Expected Breaking Point (Section~\ref{sec:expected-breaking-point})}
\label{app:expected-breaking-point}

\lemmaExpectedBPStrategy*
\begin{proof}
	\begin{itemize}
		\item $\expectedbreakpoint^T_{\game, \phi}(\pi)$: 
			From the cost function $C$, we get that for each run $\rho$, the number of disturbances is equal to the cost of that run, i.e.
			\[C_\rho= D^T(\rho)\]
			Therefore, the expected cost of a run $\rho$ in $M_\pi^\mu$ is equal to the expected number of disturbances.
			\[\expected_{M_\pi^\mu,s_0}(C_\rho)= \expected_{M_\pi^\mu,s_0}(D^T(\rho))\]
			\[\inf_{\mu\in U_\pi}\expected_{M_\pi^\mu,s_0}(C_\rho)= \inf_{\mu\in U_\pi}\expected_{M_\pi^\mu,s_0}(D^T(\rho))\]
			\[\expectedbreakpoint^T_{\game,\objective}(\strategyone) = \inf_{\mu\in U_\pi} TR(M_\pi^\mu, s_0, C)\]

			Hence, the minimum expected cost to violate $\phi$ is equal to the minimum expected number of disturbances required to violate $\phi$. 
			
		\item $\expectedbreakpoint^F_{\game, \phi}(\pi)$: 
			This is similar to the previous case, but here we use the equivalence of mean payoff on a run to the frequency of disturbances on that run and get the following:
			\[\expectedbreakpoint^F_{\game,\objective}(\strategyone) = \inf_{\mu\in U_\pi} MP(M_\strategyone^\mu,s_0, C)\qedhere\]
	\end{itemize}
\end{proof}

\subsection{{\new{Solving Stochastic Shortest Path Problems}}}
\label{app:SSP}
In its usual setting, the SSP assumes that the probability of reaching the target states is 1.
This assumption is required because if the run reaches any other end component with any positive probability where the minimum cost is infinite, the total expected cost becomes infinite. 
To handle cases where probability to reach the required states is not one, several other criteria have been explored in the literature, for example, \emph{SSP with Unavoidable Dead Ends (iSSPUDE)}~\cite{iSSPUDE} and \emph{Minimum Cost Given Maximum Probability (MCMP)} \cite{MCMP}.
In iSSPUDE, the probability to reach the target is maximized with higher priority and the cost is minimized later.
However, in this setting, only the costs for paths that reach the target states count. 
Other paths, that do not reach the target states are not counted, or equivalently they are assigned zero cost.
In MCMP, the priority is still maximizing the probability however the cost for runs that do not reach the target state is also counted up until the point when it is not possible to reach the target anymore. 
The intuition behind this is that once the run reaches these end components, the play is stopped and more cost does not accumulate.
This is handled by allowing a \emph{giving up} action that ends the run once it reaches a state that cannot reach the target anymore, making sure that no more cost is accumulated.
Our setting is quite similar to MCMP, since violating $\phi$ is a hard constraint and therefore prioritized. 
Also, the giving up action is always present since there is always a non-disturbance action with zero cost from every state, defined by the Player~1 strategy $\pi$. 
This can ensure zero cost runs once the run reaches end components from which it is not possible to reach the target anymore.

To solve the MCMP case, solutions using Linear Programs have been explored in \cite{SSP-tradeoff}. 
To write the linear program, we first need to define variables $x_{s,a}$ that represents the probability of picking action $a$ from state $s$. 
These variables are used to define the flow of probability mass. 
The constraints of the linear program are defined using flow equations and the objective is to minimize the cost function. 
The inward and outward flows of a state are defined as: \[in(s) = \sum_{s'\in \bar{S}, a\in \bar{A}} x_{s',a}P(s,a,s') \quad \quad out(s) = \sum_{a\in \bar{Av}(s)} x_{s,a}\].

The linear program is defined as follows:
\begin{align}
	\min_{x_{s,a}} & \sum_{s\in \bar{S}, a\in \bar{A}} x_{s,a}C(s,a) &\\
	\text{s.t.} & x_{s,a} \geq 0 & \forall s\in \bar{S}, a\in\bar{A} \\
	& out(s) - in(s) \leq 0 & \forall s\in \bar{S}\setminus(B\cup\{s_0\})\\
	& out(s_0) - in(s_0) \leq 1 & \\
	& \sum_{s_b\in B} in(s_b)\geq 1-p &
\end{align}

{\new{Note that this gives the values of $x_{s,a}$ in $[0,1]$ which are the probability of picking action $a$ from state $s$. Thus, the resulting policy is randomized.}}

\subsection{Weighted MEC Quotient}
\label{app:weighted_MEC_quotient}
\begin{figure*}[h]
	\centering
	\includegraphics[scale=0.35]{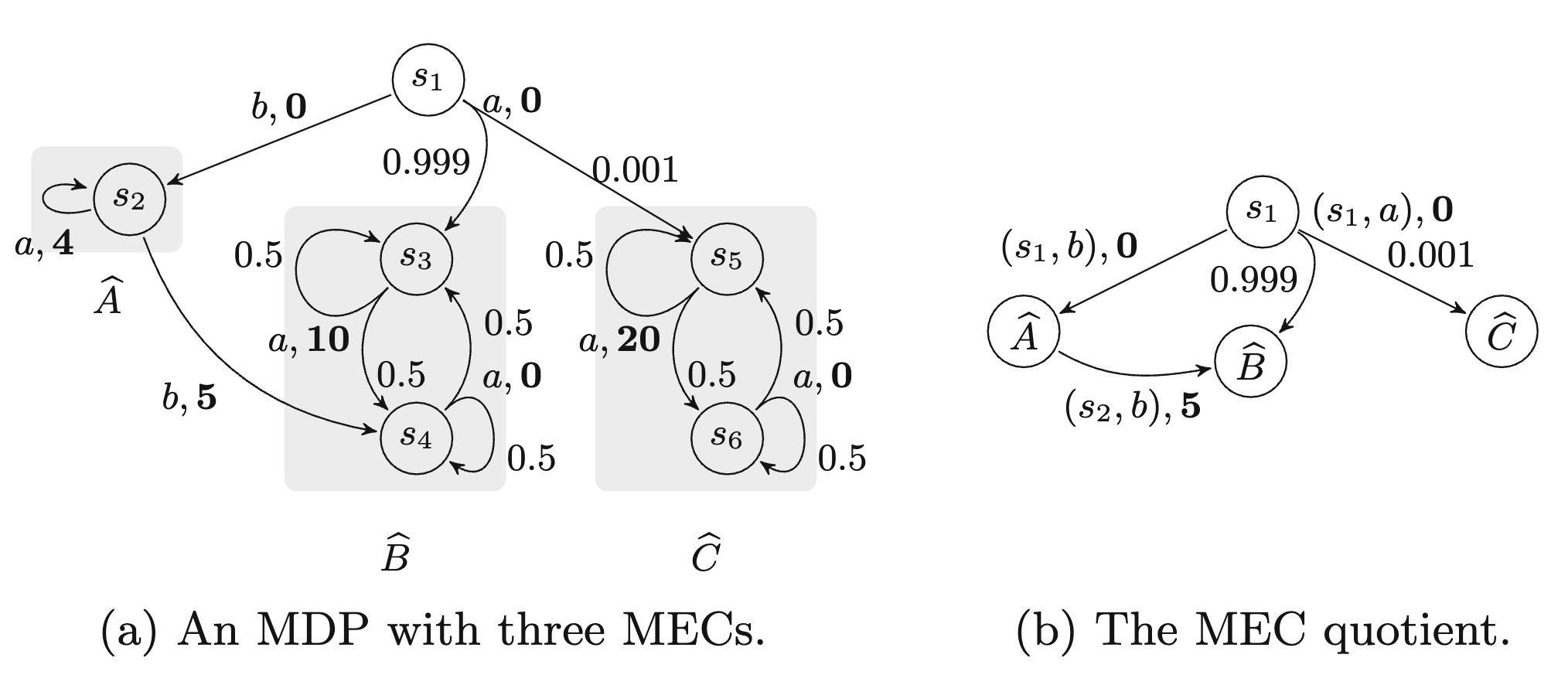}
	\caption{An example of how the MEC quotient is constructed. By $a, \mathbf{r}$ we denote that the action $a$ yields a cost of $\mathbf{r}$}
	\label{fig:MEC_quotient}
\end{figure*}

Let $M$ be an MDP with MECs $MEC(M) = {(Y_1, Z_1), \dots, (Y_l, Z_l)}$.
Further, define $MEC_S =\bigcup_{i=1}^k Y_i$ as the set of all states contained in some MEC.
The MEC quotient of $M$ is defined as the MDP $\hat{M} = (\hat{S}, \hat{A}, \hat{Av}, \hat{T})$, where:

\begin{itemize}
	\item $\hat{S} = S\setminus MEC_S \cup \{\hat{s_1}, \dots \hat{s_l}\}$,
	\item $\hat{A} = \{(s,a)\mid s\in \gamestates, A\in \gameavailactions(s)\}$,
	\item $\hat{Av}(s) = \\
		\begin{cases}
			\{(s,a)\mid a\in Av(s)\} & \forall s \in S\setminus MEC_S \\
			\{(s, a) \mid s \in Y_i, a \in Av(s)\setminus Z_i\} & \forall 1 \leq i\leq l
		\end{cases}$,
	\item $\hat{T}(\hat{s}, (s,a), \hat{s'}) = 
		\begin{cases}
			\sum_{s'\in Y_j} T(s,a,s') & \hat{s'}=\hat{s_j} \\
			T(s,a,\hat{s'}) & \hat{s'}\in S\setminus MEC_S
		\end{cases}$
\end{itemize}
If $s_0$ is in some $Y_i$, then $\hat{s_i}$ is the initial state of $\hat{M}$, otherwise $s_0$ is the initial state of $\hat{M}$.

\begin{figure}[h]
	\centering
	\includegraphics[scale=0.4]{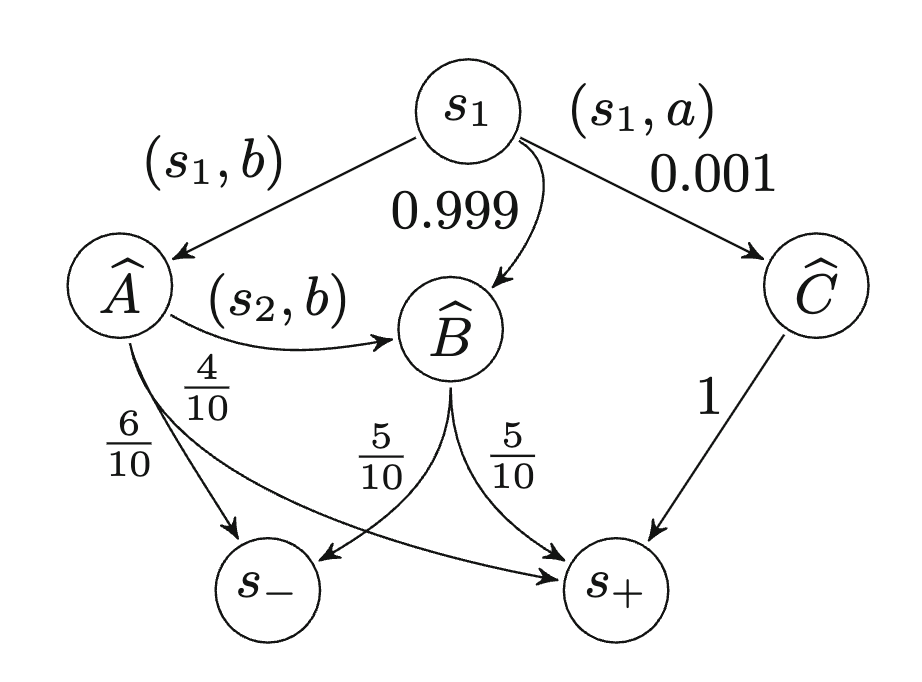}
	\caption{The weighted MEC quotient of the MDP in \Cref{fig:MEC_quotient}(a) and function $f =\{\hat{A} \rightarrow \frac{4}{10}, \hat{B} \rightarrow \frac{5}{10}, \hat{C} \rightarrow \frac{10}{10} \}$. 
	$stay$ action labels omitted for readability.}
\end{figure}

Let $\hat{M}$ bet the MEC quotient of $M$ and let $MEC_{\hat{S}} = {\hat{s_1},\dots \hat{s_n}}$ be the set of collapsed states. 
Further, let $f : MEC_{\hat{S}} \rightarrow [0, 1]$ be a function assigning a value to every collapsed state. 
We define the weighted MEC quotient of $M$ as the MDP $M^f = (S^f,\hat{A} \cup \{stay\}, Av^f , T^f)$, where
\begin{itemize}
	\item $S^f = \hat{S} \cup \{s_+\}$,
	\item $Av^f(\hat{s})$ is defined as
		\[\forall \hat{s}\in \hat{S}.\ Av^f(\hat{s})=
		\begin{cases}
			\hat{Av}(\hat{s}) \cup \{stay\}& \hat{s}\in MEC_{\hat{S}} \\
			\hat{Av}(\hat{s}) & \text{otherwise}
		\end{cases}\]
		\[Av^f(s_+)=\emptyset\]
	\item $T^f$ is defined as
		\begin{align}
		\forall \hat{s}\in \hat{S}, \hat{a}\in \hat{A}&. T^f(\hat{s}, \hat{a}) = \hat{T}(\hat{s},\hat{a}) \\
		\forall \hat{s_i}\in MEC_{\hat{S}}&. T^f(\hat{s_i}, stay, s_+) = 1 \\
		\end{align}
\end{itemize}

The new cost function is defined as $C^f(\hat{s_i}, \{stay\}) = MP(Y_i,Z_i)$ for all $\hat{s_i}\in MEC_{\hat{S}}$ and $0$ otherwise. 
Here $MP(Y_i,Z_i)$ is the minimum mean payoff of the MEC $(Y_i,Z_i)$.

\begin{lemma}
	\label{lem:weighted_MEC_quotient}
	The minimum expected cost of reaching $B\cup \{s_+\}$ in the weighted MEC quotient $M^f$ is equal to the minimum mean payoff of violating $\phi$ in the MDP $M$.

\end{lemma}
\begin{proof}
	Correctness of this lemma follows from the correctness of the transformation in~\cite{ashok17meanpayoff}.
	We can show that for every strategy $\mu^f$ in $M^f$ and strategies for MECs $\mu^{Y_i}$, there is a strategy $\mu$ in $M$ that has the mean payoff equal to the expected cost of reaching $B\cup \{s_+\}$ in $M^f$.
	Using this equivalence, we can take minimum over all strategies $\mu \in U_\pi$ that gives us the expected frequency breaking point.

\end{proof}

\subsection{Computing the Expected Breaking Point (\Cref{sec:expected_breaking_point_strategy})}

\begin{algorithm}[ht]
	\small
	\caption{Computing Expected Breaking Point of a Player~1 Strategy}
	\label{alg:expected}
	\begin{algorithmic}[1]
		\State \textbf{Input:} $\game, \phi, \strategyone$
		\State \textbf{Output:} $(\expectedbreakpoint^T_{\game, \phi}(\strategyone), \expectedbreakpoint^R_{ \game, \phi}(\strategyone))$
		\State $M_\strategyone \gets InducedMDP(\game, \strategyone)$ 
		
		\If{$\phi = P_{> p} (\square \neg B)$} \Comment{Safety}
			\State $B\gets$ Input from $\phi$ 
		\ElsIf{$\phi = P_{> p} (\lozenge \targetstates)$} \Comment{Reachability}
			\State $B\gets $ Compute Using \cref{eq:states_B}
			
		\EndIf

		\State $V_{s} \gets$ Reachability-LP($M_\strategyone, P_{\geq 1-p} (\lozenge B)$)
		\If{$V_{s_0} \geq 1-p$}
			\State $\expectedbreakpoint^T \gets$ SSP-LP($M_\strategyone, B, 1-p$)
			\State \Return $(\expectedbreakpoint^T, 0)$
			
		\ElsIf{$\phi = P_{> p} (\square \neg B)$} \Comment{safety $\implies$ unbreakable}
			\State \Return $(\natural, \natural)$
		\Else \Comment{reachability $\implies$ frequency}
		
			\State for all $s\in B, fr(s)\gets 0$
			\State $R\gets$ Compute $R$ Using \cref{eq:states_R}
			\State $V'_{s_0} \gets$ Reachability-LP($M_\strategyone, P_{\geq 1-p} (\lozenge (B\cup R))$)
			\If{$V'_{s_0} < 1-p$} \Comment{unbreakable with $R$ too}
				\State \Return $(\natural, \natural)$
			\EndIf
			\For{$r\in R$}
				\State $fr(r)\gets$ MinMeanPayoff($r$)
			\EndFor
			\State $M_\strategyone^q \gets$ ComputeWeightedQuotientMDP($M_\strategyone,R,fr$)
			\State $s_+ \gets$ NewTerminalState($M_\strategyone^q$)
			\State $\expectedbreakpoint^F \gets$ SSP-LP($M_\strategyone^q, B\cup \{s_+\}, 1-p$)
			\State \Return $(\omega, \expectedbreakpoint^F)$
		\EndIf
	\end{algorithmic} 
\end{algorithm}

\theoremExpectedBPStrategyReach*
\begin{proof}

	\textbf{Safety:} The algorithm terminates because all the steps terminate and there are no infinite loops present.
	For the correctness of this case, we need to show that if the maximum probability of reaching $B$ is $<1-p$, it is not breakable and if it is $\geq 1-p$, the breaking point is finite and equal to the solution of SSP.
	
	\begin{itemize}
		\item \textbf{Case 1: }$V_{s_0}\geq 1-p$
		
		In this case, the expected cost to reach $B$ with $\geq1-p$ probability becomes finite. 
		Minimum expected cost here can be computed using the SSP for MDP $M_\pi$ for the states $B$ giving us the transient breaking point of $\pi$.	
		Since the transient breaking point is finite, the frequency breaking point automatically becomes 0.	
		\item \textbf{Case 2: }$V_{s_0}<1-p$ \\
		Note that in this case, it is not even possible for Player~2 to reach $B$ with the required probability, therefore the strategy is not breakable and the set of strategies $U_\pi$ is empty.
	\end{itemize}
	\textbf{Reachability:} Here too, the algorithm terminates because all the steps terminate and there are no infinite loops present.
	We first compute the set of winning states $B$ for Player~2 using \cref{eq:states_B}. 
	Note that Player~1 cannot reach $G$ once the play reaches $B$.
	\begin{itemize}
		\item \textbf{Case 1: }$V_{s_0}\geq 1-p$ \\
		If the probability of reaching $B$ is $\geq 1-p$, the expected cost to reach $B$ with $\geq1-p$ probability becomes finite.
		Minimum expected cost here can be computed using the SSP for MDP $M_\pi$ for the states $B$ giving us the transient breaking point of $\pi$.		
		Since the transient breaking point is finite, the frequency breaking point automatically becomes 0.	
		\item \textbf{Case 2: }$V_{s_0}<1-p$ \\
		The probability of eventually reaching an end component is always 1.
		\[P(\lozenge (G \cup B \cup R)) = 1\]
		\[P(\lozenge G) + P(\lozenge (B \cup R)) = 1\]
		\[1 - P(\lozenge G) = P(\lozenge (B \cup R))\]
		\[P(\square \neg G) = P(\lozenge (B \cup R))\]

		If the probability of reaching $R\cup B$ is $<1-p$, it means that no strategy would violate $\phi$ and therefore $U_\strategyone=\emptyset$. The algorithm correctly returns $\natural, \natural$ in this case.
		If it is $\geq 1-p$, we need to compute the frequency breaking point which is equivalent to the minimum expected mean payoff.
		For a play to stay forever in a MEC, the minimum frequency of disturbances must be equal to the minimum mean payoff of that MEC.
		Using \cref{lem:weighted_MEC_quotient}, we can compute the minimum expected mean payoff of violating $\phi$ in the MDP $M_\strategyone$ by computing the SSP for the weighted MEC quotient $M_\strategyone^q$.\qedhere

	\end{itemize}
\end{proof}

\subsection{Optimal Strategy: Expected Breaking Point (\Cref{sec:expected-synthesis})}
\label{app:expected-SG}

\subsubsection{Transformation of SGD to SG}
We start by defining transformation of SGD to an SG as shown in \cref{fig:gadget_expected_breaking_point}. The SG is defined as $\newgame=(\newgamestates, \newgamestatesone, \newgamestatestwo, \newgameactions, \newgameavailactions,\newgametransitions)$ where 
\begin{itemize}
	\item $\newgamestatesone= \gamestatesone$,
	\item $\newgamestatestwo = \gamestatestwo\cup \{(s,a)\mid s\in\gamestatesone, a\in \gameavailactions(s)\}$,
	\item $\newgamestates=\newgamestatesone\cup\newgamestatestwo$, 
	\item $\newgameactions=\gameactions\cup \{\bot\}$,
	\item $\newgameavailactions$ is defined as\\
		$\begin{cases}
			\newgameavailactions(s)=\gameavailactions(s) & s\in\gamestates \\
			\newgameavailactions((s,a))=\{\bot,d\} & s\in\gamestatesone, a\in \gameavailactions(s)
			
		\end{cases}$
	\item $\newgametransitions$ is defined as\\
		$\begin{cases}
			\newgametransitions(s,a)=\gametransitions(s,a) & s\in\gamestatestwo, a\in\gameavailactions(s)\\
			\newgametransitions(s,a,(s,a))=1 & s\in\gamestatesone, a\in\gameavailactions(s)\\
			\newgametransitions((s,a),\bot,s') = \gametransitions(s,a,s') & s\in\gamestatesone, a\in\gameavailactions(s), s'\in\gamestates \\
			\newgametransitions((s,a),d,s') = \gametransitions(s,d,s') & s\in\gamestatesone, a\in\gameavaildisturbanceactions(s), s'\in\gamestates

		\end{cases}$
\end{itemize}

\begin{lemma}\label{lem:sg_equivalence} 
	There exist bijective functions $f$ and $g$ where 
	$f$ maps triplets $(\pi, \sigma, \delta)$ of strategies in $\game$ to pairs $(\tilde \pi, \tilde \sigma)$ of strategies in $\newgame$, 
	and $g$ maps paths $\rho$ in $\game$ to paths $\tilde \rho$ in $\newgame$, such that
	\[
	\forall \rho \in (S \times A)^* \times S, \quad \prob^{\pi, \sigma, \delta}_{\game, s_0}(\rho) = \prob^{\tilde \pi, \tilde \sigma}_{\newgame, s_0}(\tilde \rho)
	\]
\end{lemma}
\begin{proof}
	
	We define a bijection $g$ that maps a run $\rho$ in $\game$ to a run $\tilde{\rho}$ in $\newgame$.

	We construct $\tilde{\rho}$ by using the following rules for all $i\geq 0$: 
	\begin{itemize}
		\item If $s_i\in \gamestatesone$ and $a_i\in \gameavailactions(s_i)$: replace $s_i,a_i$ by $s_i,a_i,(s_i,a_i),\bot$
		\item If $s_i\in \gamestatesone$ and $a_i\in \gameavaildisturbanceactions(s_i)$: replace $s_i,a_i$ by $s_i,a_i,(s_i,a_i),d$
		
	\end{itemize}
	
	Note that this is a bijection, because every run in $\newgame$ can be mapped back to a run in $\game$ by simply removing the additional states $(s_i,a_i)$ and their following actions.

	\textbf{Forward direction:}
	For given $\strategyone, \strategytwo, \disturbancestrategy$ in $\game$, we construct equivalent strategies $\tilde{\strategyone}, \tilde{\strategytwo}$ for $\newgame$ under which every path has the same probability. 
	For some history $w$ in $\game$, let $\tilde{w}$ be the corresponding history in $\newgame$.

	$\tilde{\strategyone}(\tilde{w}\cdot s):=\strategyone(w\cdot s)$ for $s\in \newgamestatesone$.

	$\tilde{\strategytwo}$ is defined as\\
	$\begin{cases}
		\tilde{\strategytwo}(\tilde{w} \cdot s) = \strategytwo(w\cdot s) & s\in\gamestatestwo  \\
		\tilde{\strategytwo}(\tilde{w}\cdot (s,a)) = d & s\in \gamestatesone, a\in \gameavailactions(s), \disturbancestrategy(w\cdot s) \neq \bot \\
		\tilde{\strategytwo}(\tilde{w}\cdot (s,a)) = \bot & s\in \gamestatesone, a\in \gameavailactions(s), \disturbancestrategy(w\cdot s) = \bot \\

	\end{cases}
	$
	Note that the probability of any path under these strategies is the same.

	\textbf{Backward direction:}
	Given a pair of strategies $\tilde{\strategyone}, \tilde{\strategytwo}$ in $\newgame$, we construct equivalent strategies $\strategyone, \strategytwo$ for $\game$ under which every path has the same probability.
	For some history $\tilde{w}$ in $\newgame$, let $w$ be the corresponding history in $\game$.

	$\strategyone(w\cdot s):=\tilde{\strategyone}(\tilde{w}\cdot s)$ for $s\in \gamestatesone$.

	$\strategytwo(w\cdot s):=\tilde{\strategytwo}(\tilde{w}\cdot s)$ for $s\in \gamestatestwo$.

	$\disturbancestrategy(w\cdot s):=\\
	\begin{cases}
	\bot & s\in \gamestatesone, a=\strategyone(s), \tilde{\strategytwo}(\tilde{w}\cdot (s,a) = \bot) \\
	\tilde{\strategytwo}(\tilde{w}\cdot (s,a)) & s\in \gamestatesone, a=\strategyone(s), \tilde{\strategytwo}(\tilde{w}\cdot (s,a) = d)
	\end{cases}	$
	Note that the probability of any path under these strategies is the same and that these two constructions are inverses of each other proving that a bijection exists.

\end{proof}

This gives us the following corollary.
\begin{restatable}{corollary}{corollarySGEquivalence}
	
		For any set of states $A$, and strategies $\strategyone, \strategytwo, \disturbancestrategy$,
		$\prob^{\strategyone, \strategytwo, \disturbancestrategy}_{\game,s_0}(\lozenge A) = \prob^{\tilde{\strategyone}, \tilde{\strategytwo}}_{\newgame,s_0}(\lozenge A)$.
	\end{restatable}

This immediately gives us the following lemma.
\theoremExpectedBPSynthesis*

\subsubsection{Algorithm to Synthesize Optimal Strategies w.r.t Expected Breaking Point and its Correctness}

\begin{algorithm}[h]
	\caption{Compute Optimal Strategy w.r.t Expected Breaking Point}
	\label{alg:expected-synthesis}
	\begin{algorithmic}[1]
		\State \textbf{Input:} $\game, \phi$
		\State \textbf{Output:} $\expectedbreakpoint_{\game, \phi}$
		\State $\newgame \gets$ \textsc{ExpectedSG}($\game$) 

		\If{$\phi = P_{> p} (\square \neg B)$} \Comment{Safety}
			\State $B\gets$ Input from $\phi$ 
		\ElsIf{$\phi = P_{> p} (\lozenge \targetstates)$} \Comment{Reachability}
			\State $B\gets $ Compute Player~2 winning states
		\EndIf

		\State $V_{s} \gets$ ReachabilityQP($M_\strategyone, P_{\geq 1-p} (\lozenge B)$) 
		
		\If{$V_{s} < 1-p$} \Comment{breaking not possible}
			\State \Return $(\natural, \natural)$
		\Else
			\State $H\gets \emptyset$
			\State $Z\gets \mathsf{MEC}(\newgame\cap B$)
			\For{$z\in Z$}
				\State $fr(z)\gets$ MinMeanPayoff($z$)
				\If{$fr(z) = 0$}
					\State $H\gets H\cup z$
				\EndIf
			\EndFor
			\State $V^H_{s} \gets$ ReachabilityQP($M_\strategyone, P_{\geq 1-p} (\lozenge H)$) 
			\If{$V^H_{s} \geq 1-p$} \Comment{breaking with finitely many steps possible}
				\State $\expectedbreakpoint^T \gets$ SSP-SG($\newgame, H, 1-p$)
				\State \Return $(\expectedbreakpoint^T, 0)$
			\EndIf
			\State $\newgame^f \gets$ ComputeWeightedQuotientSG($\newgame,fr$)
			\State $s_+ \gets$ NewTerminalState($\newgame^f$)
			\State $\expectedbreakpoint^F \gets$ SSP-SG($\newgame^f, s_+, 1-p$)

			\State \Return $(\omega, \expectedbreakpoint^F)$
			
		\EndIf

	\end{algorithmic} 
\end{algorithm} 

\theoremExpectedBPSReach*

\begin{proof}
	
Termination of \Cref{alg:expected-synthesis} is guaranteed because all the steps terminate and there are no infinite loops present.
To prove correctness of \Cref{alg:expected-synthesis}, we analyze the two cases of safety and reachability separately.

\paragraph*{Safety:}
Based on the probability of reaching $B$ there are two cases.
\begin{itemize}
	\item \textbf{Case 1: }$V_{s_0}\geq 1-p$ \\
	In this case, the expected cost to reach $B$ with $\geq1-p$ probability becomes finite.
	To compute the minimum expected cost to reach $B$, we use the SSP for SG $\newgame$ for states $B$.
	The algorithms suggested in~\cite{patek1999stochastic}, namely, VI and PI only converge in the limit to the true value. Therefore, we can only compute an approximation of the transient breaking point.
	However, since the expected cost is finite, we can compute the frequency breaking point as 0.
	Note that as better methods become available to solve the SSP problem, we can use those methods to compute the transient breaking point.

	\item \textbf{Case 2: }$V_{s_0}<1-p$ \\
	Note that in this case, it is not even possible for Player~2 to reach $B$ with the required probability, therefore the set of strategies $U_\pi$ is empty for all $\pi$.
	Hence, it is not breakable and the algorithm correctly returns $(\natural, \natural)$.
\end{itemize}

\paragraph*{Reachability:}
In this case, first we compute the winning states $E$ for Player~2 using the backward fix-point computation.
Because of the construction, note that, Player~2 can always force the play to avoid $G$ once it reaches $E$. 
Also, since plays reach an end component with probability 1, we know that it either reaches $G$, $E$ or some other end component. 
If it reaches some other end component, Player~1 would never play actions to keep it there, therefore, only Player~2 would have the option to keep the play inside it if possible. But by construction, this end component would be in the set $E$.
Therefore, all plays would either reach $G$ or $E$ with probability 1 under optimal strategies for both players. 
Now, based on the maximum probability of reaching $E$ there are two cases.
\begin{itemize}
	\item \textbf{Case 1: }$V_{s_0}< 1-p$ \\
	In this case, it is not possible to reach $E$ with the required probability in $\newgame$ which means that even with as many disturbances actions as possible, the strategy is not breakable.
	Hence, the set of strategies $U_\pi$ is empty for all $\pi$.
	The algorithm correctly returns $(\natural, \natural)$.
	\item \textbf{Case 2: }$V_{s_0}\geq 1-p$ \\
	If the probability of reaching $E$ is $\geq 1-p$, we compute the set $B$ where the mean payoff is 0, i.e. the frequency of disturbances required is 0 in those MECs.
	
	If the probability of reaching $B$ is $\geq 1-p$, the expected cost to reach $B$ with $\geq1-p$ probability becomes finite since after reaching $B$, no disturbances are required and the expected disturbances to reach $B$ would be finite. 
	The minimum expected cost can be computed using SSP algorithm for the SG $\newgame$ for states $B$ giving us the optimal transient breaking point.
	The frequency breaking point would be 0 in this case.
	
	If the probability to reach $B$ is $<1-p$, we need to find the right MECs to stay in. 
	For this, the weighted MEC quotient MDP has cost associated with each MEC (transition to a new terminal state) equal to its mean payoff.
	Therefore, the minimum mean payoff to avoid $G$ is equal to the minimum mean payoff to reach $E$ with $\geq1-p$ probability. 
	The minimum mean payoff can be computed using SSP algorithm for the SG $\newgame^f$ for the new terminal state, giving us the optimal frequency breaking point.
	The transient breaking point is $\omega$ in this case.\qedhere
\end{itemize}
\end{proof}

\begin{added}
	Note that the solution of the SSP for SGs only provides an $\epsilon$-approximation of the true value.
	Therefore, the computed breaking points are also $\epsilon$-approximations of the true breaking points.
	However, as better methods become available to solve the SSP problem for SGs, we can use those methods to compute better approximations of the breaking points.
	Moreover, the complexity of the algorithm remains the same as that of solving the SSP for SGs.
\end{added}

\section{Worst-case Breaking Point (Section~\ref{sec:worst_case_breaking_point})}

\subsection{Difference between Realistic Worst-Case and Absolute Worst-Case}
\label{app:realistic_worst_case}
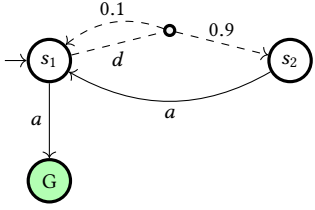
\begin{figure}[h]
	\small
	\centering
	\begin{tikzpicture}[scale=0.8]
		
		\node[state, initial, initial text=] at (0,0)     (1)   {$s_1$};
		\node[state] at (4,0)     (2)   {$s_2$};
		\node[state, fill=our-green] at (0,-2)    (3)   {G};
		\node[smallcircle] at (2,0.5)     (4)   {};
		
		\draw
		(1) edge[-, below, dashed] node {$d$} (4)
		(4) edge[->, above, bend right, dashed] node {$0.1$} (1)
		(4) edge[->, above, dashed] node {$0.9$} (2)
		(2) edge[->, bend left, below] node {$a$} (1)
		(1) edge[->, left] node {$a$} (3)
		
		;
	\end{tikzpicture}
	
	\caption{Disturbance frequency is 1 in the absolute worst case (always stay in $s_1$) but it is realized with 0 probability. Here, if the action $d$ is always played in state $s_1$, the steady state distribution is $\frac{10}{19}$ and $\frac{9}{19}$ in state $s_1$ and $s_2$ respectively. 
	This gives us the realistic worst case frequency as $\frac{10}{19}$.}
	\label{fig:recurrent_definition_justification}
\end{figure}

Here we show what happens if we consider the absolute worst case disturbance frequency instead of the realistic worst case.
For example, see \Cref{fig:recurrent_definition_justification} where the absolute worst case would give a disturbance frequency of $1$ (because of paths that always stay in $s_1$ after some finite time) however, a more useful information would be that the disturbance frequency is $\frac{10}{19}$.
We call this the realistic worst case in contrast.

\subsection{Memory Requirements for Worst-Case Breaking Point}

\worstCaseTransBP*
\begin{proof}
	($\Rightarrow$) Assume $\breakpoint^T_{\game, \phi}(\pi) \leq k$. 
	Then, by definition of breaking point, there exists a player~2 strategy $\sigma$ and a $k$-disturbance strategy $\delta$ that breaks $\pi$.
	Here, $k$-disturbance strategy means that every path under that strategy has at most $k$ disturbances. 
	
	Let $(\tilde \pi, \tilde \sigma) = f(\pi, \sigma, \delta)$, with function $f$ from \Cref{thm:unfolded_sg_equivalence} then $\sup_{\tilde \sigma}\prob^{\tilde \pi, \tilde \sigma}_{\game^{\dagger k}, (s_0,k)}(\neg \phi) \geq \prob^{\pi, \sigma, \delta}_{\game, s_0}(\neg \phi) \geq 1 - p$.

	($\Leftarrow$) Assume $\sup_{\tilde \sigma}\prob^{\tilde \pi, \tilde \sigma}_{\game^{\dagger k}, (s_0,k)}(\neg \phi) \geq 1 - p$. 
	Function $f^{-1}$ from  \Cref{thm:unfolded_sg_equivalence} gives us a disturbance strategy $\delta$ in $\game$ that breaks $\pi$ using at most $k$ disturbances, that is, $\breakpoint^T_{\game, \phi}(\pi) \leq k$.
\end{proof}

\worstCaseFreqBP*
\begin{proof}
	We start by showing that the optimal disturbance strategy that forces the run to stay inside a MEC is memoryless.

	Let $\strategytwo^*, \disturbancestrategy^*$ be the strategies that achieves the breaking point, i.e. 
	\[\breakpoint^R_{\game,\objective}(\strategyone):=
	\inf \Big\{x\in[0,1] ~\big|~ P_{\game,s_0}^{\strategyone,\strategytwo^*,\disturbancestrategy^*}\big(\rho \in Freq_{\leq x} \big) = 1\Big\}
	\]
	Now, since the cost on a path is the same as number of disturbances, we can write it as 
	\[= \inf \Big\{x\in[0,1] ~\big|~ P_{\game,s_0}^{\strategyone,\strategytwo^*,\disturbancestrategy^*}\big(\rho \in MP_{\leq x} \big) = 1\Big\}
	\]
	Where $MP_{\leq x}$ is the set of paths with mean payoff $\leq x$.

	Let $e$ be the minimum expected mean payoff. The worst case will always be greater than the expected value.
	
	$\breakpoint^R_{\game,\objective}(\strategyone)\geq e$

	Let $e$ is achieved by a memoryless strategy $\mu^*$.
	Using Lemma, we can find equivalent $\strategytwo^*,\disturbancestrategy^*$ that achieves $e$.

	Now, 
	\[P_{\game,s_0}^{\strategyone,\strategytwo^*,\disturbancestrategy^*}\big(MP(\rho) = e \big) = 1\]
	\[P_{\game,s_0}^{\strategyone,\strategytwo^*,\disturbancestrategy^*}\big(\rho \in MP_{\leq e} \big) = 1\]
	This gives us that 
	\[\breakpoint^R_{\game,\objective}(\strategyone) \leq e\]	
	
	We now show that since almost all runs reach an end component, the optimal strategy is to choose the right MECs where you always stay and maximize the probability of reaching those MECs.

	Let's assume our algorithm returns $x$ as the breaking point and $C^*$ be the set of end components that have frequency $\leq x$. 
	This implies all the remaining MECs in $C\setminus C^*$ have frequency $\geq x$.
	It also means that the algorithm produces a value that can be achieved, i.e. 
	$Output \geq \breakpoint^R_{\game,\objective}(\strategyone)$.

\end{proof}

\subsection{Computing the Worst-Case Breaking Point (\Cref{sec:worst_case_breaking_point_strategy})} 

\begin{algorithm}[H]
	\small
	\caption{Worst-case Breaking Point of a Strategy}
	\label{alg:breaking_point_worst_case}
	\begin{algorithmic}[1]
		\State \textbf{Input:} $\game, \phi, \strategyone$
		\State \textbf{Output:} $(\breakpoint^T_{\game, \strategyone}, \breakpoint^F_{\game, \strategyone})$
		\State $M_\strategyone \gets$ InducedMDP($\game, \strategyone$) 

		\If{$\phi = P_{> p} (\square \neg B)$} \Comment{Safety}
			\State $B\gets$ Input from $\phi$ 
		\ElsIf{$\phi = P_{> p} (\lozenge \targetstates)$} \Comment{Reachability}
			\State $B\gets $ Compute Using \cref{eq:states_B}
			
		\EndIf

		\State $V_{s} \gets$ UseReachabilityLP($M_\strategyone, B, 1-p$) \Comment{Eq. \eqref{lp_inf}}
		\If{$V_{s_0}>1-p$} \Comment{\textbf{Case 1}} \label{alg:main_finite}
			\State $k \gets$  ComputeUpperBoundTransient($M_\pi,B, 1-p$) 
			\State $\breakpoint^T =$ ComputeTransient($M_\pi,B,1-p,k$)
			\State $\breakpoint^F = 0$
		\ElsIf{$V_{s_0} = 1-p$}
			\State $\breakpoint^T =$ ComputeTransient($M_\pi,B,1-p,|\gamedisturbancetransitions|$)
			\If{$\breakpoint^T \neq \oslash$} \Comment{\textbf{Case 2(a)}}
				\State \Return ($\breakpoint^T, 0$)
				
			\Else \Comment{\textbf{Case 2(b)}}
				\State \Return $(\omega,0)$
				
			\EndIf
		\Else	\Comment{\textbf{Case 3}}
		
			\If{$\phi = P_{> p} (\square \neg B)$} \Comment{Safety $\implies$ Not breakable}
				\State \Return $(\natural, \natural)$

			\Else
				\State $C\gets \mathsf{MEC}(M_\pi)\setminus(G\cup B)$ \Comment{Other MECs}
				\State $V'_{s} \gets$ UseReachabilityLP($M_\strategyone, C\cup B, 1-p$)
				\If{$V'_{s_0}\geq 1-p$} \Comment{Need infinite disturbances}
					
					\State $\breakpoint^F \gets$ ComputeFrequency($M_\pi, G,B,1-p$)
					\State \Return $(\omega, \breakpoint^F)$
				\EndIf
			\EndIf
		\EndIf

		\State \Return $(\breakpoint^T, \breakpoint^F)$
	\end{algorithmic} 
\end{algorithm}

We first prove that if $\phi$ is reachability and it is possible to break the strategy in finitely many disturbances, then we still need to solve reachability for the set $B$ (as computed in \Cref{eq:states_B}), i.e.
\begin{lemma}
	If $\phi$ is reachability and $\strategyone$ is breakable with finite disturbances then there exists $\mu$ such that
	$\prob_{M_\strategyone, s_0}^{\mu}(\square \lnot G) = \prob_{M_\strategyone, s_0}^{\mu}(\lozenge B)$.
\end{lemma}
\begin{proof}

	We compute the set $B$ using \cref{eq:states_B} and we know that $G$ is not reachable from $B$.
	Therefore, the probability of staying out of $G$ is equal to the probability of reaching $B$ or the probability of reaching some other end component where the play can stay forever.
	Under any strategy, the probability of reaching a MEC and staying there is 1 \cite{puterman2014markov},
	
	$P_{M_\strategyone, s_0}^{\mu}(\lozenge G \lor \lozenge O \lor \lozenge B) = 1$
	Since $G$ and $B$ are sink states/end components.

	We also compute the set $R$ using \Cref{eq:states_R}.
	From the definition of $R$, you need infinitely many disturbances to stay in $R$, this implies
	$\prob_{M_\strategyone, s_0}^{\strategytwo, \disturbancestrategy}(\lozenge R) = 0$.
	Also, since Player~1 would never want to stay in $O\setminus R$, it will take the exit action whenever the play reaches $O\setminus R$.
	This mean that the probability of staying forever in $O\setminus R$ is 0.
	Therefore, $\prob_{M_\strategyone, s_0}^{\strategytwo, \disturbancestrategy}(\lozenge G) = 1-P_{M_\strategyone, s_0}^{\strategytwo, \disturbancestrategy}(\lozenge B)$.
	Also, by definition, $\prob_{M_\strategyone, s_0}^{\strategytwo, \disturbancestrategy}(\lozenge G) = 1-P_{M_\strategyone, s_0}^{\strategytwo, \disturbancestrategy}(\square \neg G)$.\qedhere

\end{proof}

We now discuss when $V_{s_0} = 1-p$ and we need to split it into two cases because it could either mean that this probability is reached only in the limit after infinitely many disturbances, or after a finite number of disturbances.

For the following, we assume that $V^{k}(s)$ denotes the probability to reach $B$ with $k$ disturbances from state $s$ and arbitrarily many from all other states.
\begin{lemma}
	\label{lem:finite-res-equal-case} 
	If the probability of reaching $B$ is exactly $1-p$ then for all $k > 0$ it holds that $V^{k+1}(s) > V^k(s)$ implies $V^{k+2}(s) > V^{k+1}(s)$.
\end{lemma}
\begin{proof}
	Let $d$ be a disturbance action in state $s$ and let $\nu$ be the probability mass of coming back to state $s$ after picking $d$. 
	If $\nu = 0$, then $V^{k+1}(s) = V^k(s)$, a contradiction. 
	Thus, $\nu > 0$, which means that there is a path $\rho = s, a, \dots, s$ (starting and ending in $s$) with probability mass $\nu$. 
	Let $\delta'(s, k+2) = d$, and follow $\rho$ afterwards. 
	Then, $V^{k+2}(s) \geq V^{\delta'}(s) > V^{k+1}(s)$.
\end{proof}

\begin{lemma}
	\label{lem:one_edge_finiteness}
	If $V_{s_0} = 1-p$, then there is a disturbance strategy $\delta$ and a natural number $k \in \langle \omega \rangle$ such that $\delta$ breaks $\pi$ with at most $k$ disturbances if and only if there is a disturbance strategy that breaks $\pi$ using each disturbance edge at most once.
\end{lemma}
\begin{proof}
	($\Leftarrow$) This is immediate.

	($\Rightarrow$) Assume that $\delta$ breaks $\pi$ using $k \in \langle \omega\rangle$ disturbances, 
	and assume that there is a disturbance edge $d \in \gamedisturbancetransitions$ that is taken $n > 1$ times. 
	There are two cases:
	\begin{itemize}
		\item Case $V^{k} \leq V^{k-1}$. Then, define $\delta'$ to never take $d$.
		\item Case $V^{k} > V^{k-1}$. Then, $V^{k+1} > V^{k}$ by  \Cref{lem:finite-res-equal-case}, 
		meaning that there exists a disturbance strategy that achieves $V_{s_0} > p$ by using $k + 1$ disturbances, a contradiction.
	\end{itemize}
\end{proof}

These lemmas tell us that if $\pi$ cannot be broken using $|T^D|$ disturbances then it cannot be broken with finitely many disturbances.
We are now ready to prove the correctness of our algorithm.

\theoremBPWC*
\begin{proof}
	We prove the correctness of \cref{alg:breaking_point_worst_case}, which runs in $\PSPACE$.
	Since reachability is reduced to safety, we just argue in terms of probability of reaching the set $B$.
	We have four cases:

	\textbf{Case 1:} $V_{s_0} > 1-p$ \\
	Let $V_{s_0}^0, V_{s_0}^1, \dots$ is the sequence of values produced by value iteration. 
	It is a non-decreasing sequence that converges to the true value $v_{s_0} > 1-p$ in the limit (by correctness of value iteration).
	There exist $j\in \Naturals$ such that $V_{s_0}^j > 1-p$.
	This value iteration upto $j$ steps would also give a strategy $\delta: (\bar S \times \bar A)^* \times \bar{S} \rightarrow \bar{A}$ that defines which action to take from every state after each time step.
	This $\delta$ breaks $\strategyone$ in $\leq j$ steps, therefore, $j$ is an upper bound.
	
	We give this upper bound to the procedure to compute transient breaking point and for its correctness, we refer to the next subsection.

	\textbf{Case 2:} $V_{s_0} = 1-p$. \\
	As we check for the transient breaking point up until the number of disturbances $|\gamedisturbancetransitions|$, we have two subcases.
	\begin{itemize}
		\item It breaks within $|\gamedisturbancetransitions|$ disturbances, and the transient breaking point is already computed.
		We leave the correctness of computing the transient breaking point to be shown later in \Cref{thm:transient_breaking_point}.
		In this case, the frequency is 0 by definition.

		\item If it does not break within $|\gamedisturbancetransitions|$ disturbances, we know that it is not possible to break it with finitely many disturbances using \cref{lem:one_edge_finiteness}.
		
		To show correctness of computing the frequency in this case, we refer to \Cref{thm:frequency_worst_case}.
	\end{itemize}

	\textbf{Case 3:} $V_{s_0} < 1-p$ \\
	If the specification is safety, the probability to reach $B$ is $<1-p$ and therefore it would never be possible to break it and therefore the breaking point is $(\natural, \natural)$.

	If the specification is reachability, we compute the set of MECs $R$ that are not $B$ or $G$.
	Here, it is possible that some runs stay in $R$ and some runs stay in $B$ to avoid reaching $G$.
	Here, the amount of disturbances required is infinite since if it was finite, the runs would reach a state from which no more disturbances are required with probability $\geq 1-p$, since this is also the definition of $B$, probability to reach $B$ would be $\geq 1-p$, a contradiction.
\end{proof}

\subsection{Algorithm for Computing Worst-Case Transient Breaking Point and its Correctness} 
\label{app:algo_transient_worst_case}

{
	\small
	\begin{equation} 
		\label{lp_0}
		\boxed{
			\begin{aligned}
				&\min~~\sum_{s\in S}V_{s, 0} \\
				& V_{s,0} = 1 \quad \forall s \in B \\
				& V_{s,0} \geq 0 \quad \forall s \in S \\
				& V_{s,0} \geq \sum_{s' \in S} \bar{T}(s,a,s')V_{s',0} \quad \forall a = \strategyone(s) \\
			\end{aligned} 
		}		
	\end{equation}
}

{
	\small
	\begin{equation}
		\label{lp_i}
		\boxed{
			\begin{aligned}
				&\min~~\sum_{s\in S}V_{s,i}\\
				&V_{s,i} = 1 \quad \forall s \in B \\
				&V_{s,i} \geq 0 \quad \forall s \in S \\
				&V_{s,i} \geq \sum_{s' \in S} \bar{T}(s,a,s')V_{s',i} \quad \forall a = \strategyone(s) \\
				&V_{s,i} \geq \sum_{s' \in S} \bar{T}(s,a,s')V_{s',i-1} \quad \forall a \in Av(s)\setminus\strategyone(s) \\
			\end{aligned}
		}		
	\end{equation}
}

\begin{algorithm}[h]
	\caption{Compute Transient Breaking Point}
	\label{alg:transient_breaking_point}
	\begin{algorithmic}[1]
		\State Given $M_\strategyone, B, k$

		\State $V_{s,0} \gets \textsf{LP$_0$}(M_\pi, B)$ \Comment{Using Eq. \eqref{lp_0}}
		
		\State $i=0$		
		\While {$i \leq k$}
		\State $V_{s,i+1} \gets \textsf{LP$_{i+1}$}(M_\pi, B, V_{s,i})$ \quad \Comment{Using Eq. \eqref{lp_i}}
		\If {$V_{s,i+1} \geq 1-p$}
		\State \Return i+1 \Comment{breaking point is $i+1$}
		\EndIf
		\State $i++$
		
		\EndWhile
		\State \Return $\oslash$ \Comment{transient breaking point is not smaller than $k$}
	\end{algorithmic} 
\end{algorithm}

\begin{restatable}{lemma}{theoremTransientBreakingPoint}
	\label{thm:transient_breaking_point} 
	\Cref{alg:transient_breaking_point} terminates and returns $\breakpoint^T_{\game,\objective}$ if  $\breakpoint^T_{\game,\objective}\leq k$.
\end{restatable}
 \begin{proof}
 	We first prove deciding $k$-resilience (transient) in SGD would be the same as solving the unfolding of the MDP obtained by treating Player~1 states as Player~2 states (see \Cref{lem:productMdpEquivalence}). 
	This requires evaluating the system under all possible $k$-disturbance strategies of Player~2.
 	Second, we show solving this unfolding is same as solving the iterative LP (see \Cref{lem:incrementalEqualsUnfolded}).

 \end{proof}

\subsubsection{Construction of Unfolding of an MDP}

We reduce the problem of finding breaking point of strategy to solving an unfolded MDP.
In this section, we explain the construction of the unfolded MDP incorporating disturbances.
To analyze the resilience of the strategy $\strategyone$, we define the unfolded MDP as an extension of the MDP $M_\strategyone$. 
The unfolded MDP introduces a disturbance counter that tracks the number of disturbances applied. 
Moreover, no disturbance actions can be chosen when the counter reaches its maximum number of allowed disturbances $k$.

We define the unfolded MDP as $M_\strategyone^\otimes = (S^\otimes, A^\otimes, \text{Av}^\otimes, T^\otimes),
$
where:
\begin{itemize}
	\item $S^\otimes = S \times \{0, \dots, k\}$ is the augmented state space, where:
	\begin{itemize}
		\item $S = \gamestatesone \cup \gamestatestwo$ is the state space of $M_\strategyone$,
		\item $i \in \{0, \dots, k\}$ is the disturbance counter, representing the number of disturbances applied so far.
	\end{itemize}
	
	\item $A^\otimes = A = \gameactions \cup \gamedisturbanceactions$ is the action space, unchanged from $M_\strategyone$.
	
	\item ${Av}^\otimes((s, i)) =
	\begin{cases}
		\{\strategyone(s)\} \cup \gameavaildisturbanceactions(s), & \text{if } s \in \gamestatesone \text{ and } i < k, \\
		\{\strategyone(s)\}, & \text{if } s \in \gamestatesone \text{ and } i = k, \\
		\gameavailactions(s), & \text{if } s \in \gamestatestwo.
	\end{cases}$
	
	\item {\small$T^\otimes((s, i), a, (s', i')) \\ =
		\begin{cases}
			\gametransitions(s, a, s') \cdot \delta_c(i' = i), & \text{if } s \in \gamestatesone, a = \strategyone(s), \\
			\gamedisturbancetransitions(s, a, s') \cdot \delta_c(i' = i+1), & \text{if } s \in \gamestatesone, a \in \gameavaildisturbanceactions(s) \text{ and } i < k, \\
			0, & \text{if } s \in \gamestatesone, a \in \gameavaildisturbanceactions(s) \text{ and } i = k, \\
			\gametransitions(s, a, s') \cdot \delta_c(i' = i), & \text{if } s \in \gamestatestwo, a \in \gameavailactions(s).
		\end{cases}$
	}
\end{itemize}

One can solve this unfolded MDP to get the transient breaking point of $\strategyone$ if it is smaller than $k$.

\para{LP for unfolded MDP.} 
Now the LP to solve the unfolding MDP can be give as follows.
\begin{equation} \label{lp_val_k_disturbance}
	\boxed{
	\begin{aligned}
		\min & \quad \sum_{(s,i) \in S^\otimes} \textsf{Val}_{(s,i)} \quad \text{subject to:} \\
		\textsf{Val}_{(s,i)} &= 1, \quad \text{for all } (s, i) \in B^\otimes; \\
		\textsf{Val}_{(s,i)} &\geq 0, \quad \text{for all } (s, i) \in S^\otimes; \\
		\textsf{Val}_{(s,i)} &\geq \sum_{s' \in S} \gametransitions(s, \pi(s), s') \textsf{Val}_{(s',i)}, 
		\quad \forall (s, i) \in S^\otimes, s \in \gamestatesone; \\
		\textsf{Val}_{(s,i)} &\geq \sum_{s' \in S} \gamedisturbancetransitions(s, a, s') \textsf{Val}_{(s',i-1)}, \\ 
		& \quad \quad \quad \quad \quad \quad \quad \forall (s, i) \in S^\otimes, a \in \gameavaildisturbanceactions(s), 1 < i \leq k; \\
		\textsf{Val}_{(s,k)} &\geq \sum_{s' \in S} \gametransitions(s, \pi(s), s') \textsf{Val}_{(s',k)}, 
		\quad \forall (s,k) \in S^\otimes, s \in \gamestatesone; \\
		\textsf{Val}_{(s,i)} &\geq \sum_{s' \in S} \gametransitions(s, a, s') \textsf{Val}_{(s',i)},\\ 
		& \quad \quad \quad \quad \quad \quad \forall (s, i) \in S^\otimes,
		a \in \gameavailactions(s), s \in \gamestatestwo.
	\end{aligned}
}
\end{equation}

\begin{lemma}
	\label{lem:productMdpEquivalence}
	Given a set of states $B \subseteq \gamestates$, a strategy with memory $\mu : S \times \{0, \dots, k\} \to A$ for the MDP $M_\strategyone$, 
	
	there exists a strategy $\mu^\otimes$ for the unfolded MDP $M_\strategyone^\otimes$ such that:
	\[
	P_{M_\strategyone, (s_0, 0)}^{\mu, \delta}(\lozenge B) = P_{M_\strategyone^\otimes, (s_0, 0)}^{\mu^\otimes}(\lozenge B^\otimes),
	\]
	Here, $B^\otimes = \{(s, i) \mid s \in B, i \in \{0, \dots, k\}\}$ is the corresponding set of target states in the unfolded MDP.
\end{lemma}

\begin{proof}
	We aim to show that for a given strategy in one MDP, 
	we can construct a strategy in the other MDP that ensures equivalent probabilities of reaching the target set of states.
	
	\para{Constructing a strategy in the unfolded MDP}
	Let $\mu: S \times \{0, \dots, k\} \to A$ be the memory-based strategy for Player 1 in $M_\strategyone$. 
	We construct a corresponding strategy $\mu^\otimes$ for the unfolded MDP $M_\strategyone^\otimes$ as follows:
	\begin{itemize}
		\item For each state $(s, i) \in S^\otimes$, define $\mu^\otimes((s, i)) = \mu(s, i)$.
	\end{itemize}
	This strategy ensures that the decisions made in the unfolded MDP align with those in the original MDP under the strategy $\mu$.
	
	\para{Reaching the target set $B$.}
	Suppose there exists a strategy $\mu$ in $M_\strategyone$ such that the probability of reaching $B \subseteq \gamestates$ from $(s_0, 0)$ is $P$. That is,
	\[
	P_{M_\strategyone, (s_0, 0)}^{\mu, \delta}(\lozenge B) = P.
	\]
	We claim that under the constructed strategy $\mu^\otimes$, the same probability $P$ is achieved in the unfolded MDP for reaching the corresponding target set $B^\otimes$. 
	Specifically,
	\[
	P_{M_\strategyone^\otimes, (s_0, 0)}^{\mu^\otimes}(\lozenge B^\otimes) = P.
	\]
	
	The equivalence follows from the fact that the dynamics of the unfolded MDP are designed to mirror the augmented state space of the original MDP. 
	In particular:
	\begin{itemize}
		\item The states in $S^\otimes$ correspond exactly to the augmented states $S \times \{0, \dots, k\}$ of $M_\strategyone$.
		\item The transitions in $M_\strategyone^\otimes$ under $\mu^\otimes$ are equivalent to the transitions in $M_\strategyone$ under $\mu$.
		\item The target set $B^\otimes$ corresponds directly to $B$ in the original MDP, extended with all memory labels $i \in \{0, \dots, k\}$.
	\end{itemize}
	Since the strategy $\mu^\otimes$ replicates the decisions of $\mu$ and the structure of $M_\strategyone^\otimes$ mirrors the augmented state space of $M_\strategyone$,
	the probability of reaching $B^\otimes$ in $M_\strategyone^\otimes$ under $\mu^\otimes$ matches the probability of reaching $B$ in $M_\strategyone$ under $\mu$.
	
	\para{Constructing a strategy in the original MDP.}
	Conversely, suppose there exists a strategy $\mu^\otimes$ in $M_\strategyone^\otimes$ that ensures a probability $P$ of reaching $B^\otimes$ from $(s_0, 0)$. 
	That is,
	\[
	P_{M_\strategyone^\otimes, (s_0, 0)}^{\mu^\otimes}(\lozenge B^\otimes) = P.
	\]
	We can construct a corresponding strategy $\mu$ in $M_\strategyone$ by defining:
	\begin{itemize}
		\item $\mu(s, i) = \mu^\otimes((s, i))$ for all $(s, i) \in S^\otimes$.
	\end{itemize}
	By similar reasoning, this strategy $\mu$ ensures that the probability of reaching $B$ in $M_\strategyone$ matches the probability of reaching $B^\otimes$ in $M_\strategyone^\otimes$ under $\mu^\otimes$.
\end{proof}

\begin{lemma}\label{lem:incrementalEqualsUnfolded}
	Solving the incremental LP is equivalent to solving the unfolded MDP.
\end{lemma}
\begin{proof}
	We prove, by induction on $k$, that the minimized sum of the objective in the incremental LP in Eq.~\eqref{lp_i} coincides with the minimized sum in the unfolded MDP's LP in Eq.~\eqref{lp_val_k_disturbance}. Specifically, we aim to show:
	\[
	\min \sum_{s \in S} V_{s,k} = \min \sum_{(s,i) \in S^\otimes} \textsf{Val}_{(s,i)}.
	\]

	\para{Base Case.} 
	For $k = 0$, the incremental LP (Eq.~\eqref{lp_0}) minimizes:
	\[
	\sum_{s \in S} V_{s,0}, \quad \text{subject to:}
	\]
	$\begin{aligned}
	V_{s,0} &= 1 \quad \text{for all } s \in B;\\
	V_{s,0} &\geq 0, \quad V_{s,0} \geq \sum_{s' \in S} \bar{T}(s, \pi(s), s') V_{s',0}.
	\end{aligned}
	$
	
	The unfolded MDP LP for $k = 0$ minimizes:
	\[
	\sum_{(s,0) \in S^\otimes} \textsf{Val}_{(s,0)}, \quad \text{subject to:}
	\]
	$\begin{aligned}
	\textsf{Val}_{(s,0)} &= 1 \quad \text{for all } (s, 0) \in B^\otimes;\\ \textsf{Val}_{(s,0)} &\geq 0, \quad \textsf{Val}_{(s,0)} \geq \sum_{s' \in S} \gametransitions(s, \pi(s), s') \textsf{Val}_{(s',0)}.
\end{aligned}
$
	
	The constraints and objective in \Cref{lp_0} and the unfolded MDP LP are identical when $k = 0$, as both minimize the sum over the same state values under the same set of constraints.
	Therefore:
	\[
	\min \sum_{s \in S} V_{s,0} = \min \sum_{(s,0) \in S^\otimes} \textsf{Val}_{(s,0)}.
	\]

	\para{Induction Hypothesis.}
	Assume that for $k \geq 0$, the minimized sums in the incremental LP and the unfolded MDP LP are equal:
	\[
	\min \sum_{s \in S} V_{s,k} = \min \sum_{(s,i) \in S^\otimes, \, i \leq k} \textsf{Val}_{(s,i)}.
	\]

	\para{Inductive Step} 
	We now prove that the equivalence holds for $k+1$. The incremental LP for $k+1$ minimizes:
	\[
	\sum_{s \in S} V_{s,k+1}, \quad \text{subject to:}
	\]
	\[
	V_{s,k+1} = 1 \quad \text{for all } s \in B,
	\]
	\[
	V_{s,k+1} \geq 0,
	\]
	\[
	V_{s,k+1} \geq \sum_{s' \in S} \bar{T}(s, \pi(s), s') V_{s',k+1},
	\]
	\[
	V_{s,k+1} \geq \sum_{s' \in S} \bar{T}(s, a, s') V_{s',k}, \quad \text{for all } a \in Av(s) \setminus \pi(s).
	\]
	
	The unfolded MDP LP for $k+1$ minimizes:
	\[
	\sum_{(s,i) \in S^\otimes, \, i \leq k+1} \textsf{Val}_{(s,i)}, \quad \text{subject to:}
	\]
	\[
	\textsf{Val}_{(s,i)} = 1 \quad \text{for all } (s, i) \in B^\otimes,
	\]
	\[
	\textsf{Val}_{(s,i)} \geq 0,
	\]
	\[
	\textsf{Val}_{(s,k+1)} \geq \sum_{s' \in S} \gametransitions(s, \pi(s), s') \textsf{Val}_{(s',k+1)},
	\]
	\[
	\textsf{Val}_{(s,k+1)} \geq \sum_{s' \in S} \gamedisturbancetransitions(s, a, s') \textsf{Val}_{(s',k)}, \quad \text{for all } a \in \gameavaildisturbanceactions(s).
	\]

	\paragraph*{Matching the Minimized Sums}
	
	\begin{itemize}
		\item Structure of the Recursion:
		Both the incremental LP and the unfolded MDP LP minimize over values that depend on contributions from two sources:
		\begin{itemize}
			\item Values propagated via $\bar{T}(s, \pi(s), s')$ (Player~1’s strategy).
			\item Values propagated via $\bar{T}(s, a, s')$ or $\gamedisturbancetransitions(s, a, s')$ (Player~2’s disturbances).
		\end{itemize}
		These contributions are identical in both LPs, with the unfolded MDP LP explicitly encoding the disturbance level $i$, and the incremental LP aggregating these contributions.
		
		\item 
		By the inductive hypothesis:
		\[
		\min \sum_{s \in S} V_{s,k} = \min \sum_{(s,i) \in S^\otimes, \, i \leq k} \textsf{Val}_{(s,i)}.
		\]
		Substituting this into the constraints for $k+1$, we see that the constraints enforce the same balance between Player~1’s strategy and Player~2’s disturbances in both LPs.
		
		\item
		The additional terms for $k+1$ disturbances are minimized in the same way in both formulations. Therefore:
		\[
		\min \sum_{s \in S} V_{s,k+1} = \min \sum_{(s,i) \in S^\otimes, \, i \leq k+1} \textsf{Val}_{(s,i)}.\qedhere
		\]
	\end{itemize}
\end{proof}

\subsection{Algorithm for Computing Worst-Case Frequency Breaking Point and its Correctness}

\begin{algorithm}[h]
	\caption{Compute Frequency Breaking Point}
	\label{alg:frequency_worst_case}
	\begin{algorithmic}[1]
		\State \textbf{Input:} $M_\pi, G, B, 1-p$

		\State for all $s\in B, fr(s)\gets 0$
		
		\State $R \gets $ Compute $R$
		\For{$r\in R$}
			\State $fr(r)\gets \textsf{ComputeFrequencyMEC}(r)$
		\EndFor
		\State $Z \gets R$
		
		\While{\textbf{true}}
			\State $c\gets \textsf{MaxFreqMEC}(Z)$
			\State $Z\gets Z\setminus \{c\}$
			\If{$P(\lozenge Z) < 1-p$}
				\State \Return $fr(c)$
			\EndIf
		\EndWhile
	\end{algorithmic} 
\end{algorithm}

\begin{restatable}{lemma}{theoremFrequencyWorstCase}
	\label{thm:frequency_worst_case}
	\Cref{alg:frequency_worst_case} terminates and returns $\breakpoint^F_{\game,\objective}(\strategyone)$ in polynomial time.
\end{restatable}

\begin{proof}
	
	From the definition of frequency breaking point
	\[\breakpoint^R_{\game,\objective}(\strategyone) = \inf_{\strategytwo,\disturbancestrategy \in U_\strategyone} \Big\{x ~\big|~ P_{\game,s_0}^{\strategyone,\strategytwo,\disturbancestrategy}\big(\rho \in Freq_{\leq x} \big) = 1\Big\}\]
	From \Cref{lem:induced_mdp_equivalence}, we get
	\[P_{\game,s_0}^{\strategyone,\strategytwo,\disturbancestrategy}\big(\rho \in Freq_{\leq x} \big) = P_{M_\strategyone,s_0}^{\strategytwo,\disturbancestrategy}\big(\rho \in Freq_{\leq x} \big)\]

	From here onwards, we drop the sub- and superscript of $P$ for ease of notation
	\[ P \big(\rho \in Freq_{\leq x} \big) = P \big(\rho \in MP_{\leq x} \big) \]
	where, $MP_{\leq x}$ is the set of runs that have mean payoff $\leq x$.

	Let $C=\{C_1,C_2,\dots C_m\}$ are the MECs of $M_\strategyone$.
	Under memoryless strategies $\strategytwo,\disturbancestrategy$, the probability of a path reaching a MEC and staying there forever is 1. 
	This gives us,
	\[=P \Big( \rho \in MP_{\leq x} \bigcap \big(\bigcup_i \lozenge C_i \big) \Big)\]
	Where $\lozenge C_i$ denotes the set of runs eventually reaching the MEC $C_i$ and staying there forever.
	Now, since all the MECs are disjoint, we get
	\[=\sum_{i} P \big( \rho \in MP_{\leq x} \cap \lozenge C_i \big)\]

	For some $C_i$, 
	$P \big( \rho \in MP_{\leq x} \cap \lozenge C_i \big) > 0$ only if the mean payoff of $C_i$ is $\leq x$.

	\[=\sum_i P \big( \rho \in \lozenge C_i\big) [MP(C_i)\leq x] \]
	where $[]$ is the bracket notation and $MP(C_i)$ is the mean payoff of $C_i$.

	\[=P \big( \rho \in \{\lozenge C_i \mid MP(C_i)\leq x \}\big)\]

	Since this probability should be 1 according to the definition, almost all the runs should reach only the end components that have mean payoff $\leq x$.
\end{proof}

\subsection{Optimal Strategy: Worst-Case Breaking Point (\Cref{sec:worst_case_synthesis})}

\paragraph{Reachability.}

Initially, we also assume that the SGD has a special structure: 
we have three assumptions on the structure of the SGD:
\begin{itemize}
	\item (A1.) For each state $s\in \gamestates$, $|\gameavailactions(s)|=2$.
	\item (A2.) For each state $s\in \gamestatesone$, $|\gameavaildisturbanceactions(s)|=1$. 
	\item (A3.) For each state,for every action, with a small probability $\epsilon$,
	the game leads to $B$.
\end{itemize}
The intuition behind it is that, with these restrictions, its $k$-unfolding will be a special
type of \emph{stopping stochastic game}~\cite{condon1992complexity},
for which it is easier to formulate a QP. 
Note that \cite{condon1992complexity}
considers the probability values to be in $\{0,0.5,1\}$, but
this restriction can be ignored without any loss of generality
(as discussed in \cite[Lemma 3]{kvretinsky2022comparison}).
We discuss how to accommodate the general case later.  

\para{Iterative Quadratic Programming.}
In this section, for a state $s\in \gamestatesone$, we denote the only disturbance action in $\gameavaildisturbanceactions(s)$ as $d_s$.
We define a sequence of QPs, namely $(QP_0,QP_1,\ldots QP_k)$, where in $QP_i$ we minimize the objective function $F_i$ w.r.t. a constraint set $C_i$.
\begin{align}
	\small
	F_0 &= \sum_{s \in \gamestates}
	\prod_{a\in \gameavailactions(s)}\left(V_{s,0} - \sum_{s' \in \gamestates} \gametransitions(s,a,s')\cdot V_{s',0}
	\right)\label{eq:F0}\\\
	F_i &= 
	\sum_{s \in \gamestatesone}\prod_{a\in \gameavailactions(s)}\left(
	V_{s,i} - V_{s,i,a}
	\right)\notag\\ 
	&+
	\sum_{s \in \gamestatestwo}\prod_{a\in \gameavailactions(s)}\left(
	V_{s,i} - \sum_{s' \in \gamestates} \gametransitions(s,a,s')\cdot V_{s',i}
	\right)\notag\\
	&+
	\sum_{s \in \gamestatesone}\sum_{a\in \gameavailactions(s)}\left(
	V_{s,i,a} - \sum_{s' \in \gamestates} \gametransitions(s,a,s')\cdot V_{s',i}
	\right)\notag\\
	&\quad\quad\quad\quad\cdot \left(
	V_{s,i,a} - \sum_{s' \in \gamestates} \gametransitions(s,d_s,s')\cdot V_{s',i-1}
	\right)
\end{align}
\normalsize

	\begin{equation} \label{qp_0}
		\small
		\boxed{
			\begin{aligned}
				\min & F_0 \quad \text{subject to:}\\
				V_{s,0} &= 1\quad \forall s\in G;\\
				V_{s,0} &= 0\quad \forall s\in B;\\
				V_{s,0} &\geq \sum_{s' \in \gamestates} \gametransitions(s,a,s')\cdot V_{s',0}\\
				&\forall s \in \gamestatesone , a \in \gameavailactions(s); \\
				V_{s,0} &\leq \sum_{s' \in \gamestates} \gametransitions(s,a,s')\cdot V_{s',0}\\ & \forall s \in \gamestatestwo, a \in \gameavailactions(s). \\
			\end{aligned}
		}
		\normalsize
	\end{equation}

	\begin{equation}\label{qp_1}
		\small
		\boxed{
			\begin{aligned}
				\min & F_i \quad \text{subject to:}\\
				V_{s,i} &= 1\quad \forall s\in G;\\
				V_{s,i} &= 0\quad \forall s\in B;\\
				V_{s,i} &\geq V_{s,i,a} \quad \forall s \in \gamestatesone, a \in \gameavailactions(s); \\
				V_{s,i,a} &\leq \sum_{s' \in \gamestates} \gametransitions(s,a,s')\cdot V_{s',i} \quad \forall s \in \gamestatesone, a \in \gameavailactions(s); \\
				V_{s,i,a} &\leq \sum_{s' \in \gamestates} \gametransitions(s,d_s,s')\cdot V_{s',i-1} \quad \forall s \in \gamestatesone, a \in \gameavailactions(s); \\
				V_{s,i} &\leq \sum_{s' \in \gamestates} \gametransitions(s,a,s')\cdot V_{s',i}\quad \forall s \in \gamestatestwo, a \in \gameavailactions(s);
			\end{aligned}
		}
		\normalsize
	\end{equation}

\vspace{0.1in}

where the values of $V_{s,i-1}$ in $C_i$ comes from the solution of $QP_{i-1}$. For a state $s$, $d_s$ is the only disturbance action in $\gameavaildisturbanceactions(s)$.

\paragraph{Safety.}
For a stopping SG, for any pair of strategies, with probability $1$, eventually one of the sink state is reached. Without loss of generality, the SG has two sink states $G$ and $B$. 
Then the objective 
$\objective=P_{\geq p}(\square \neg B)$ can be converted to $\objective=P_{\leq p}(\lozenge G)$.
Then, for safety objective, the existence of a $k$-resilient strategy can be checked by checking if  $V_{s_{0},k}\leq p$.

\begin{algorithm}[t]
	\caption{Compute Worst-case Transient Breaking Point}
	\label{alg:k_resilience}
	\begin{algorithmic}[1]
		\State Given $\game, \phi, k$
		\State $i=0$
		\State $\overline{V}\gets \textsf{QP0}(\game, \phi)$\Comment{Eq. \eqref{qp_0}}
		\While {$i \leq k$}
		\State $\overline{V}\gets \textsf{QP}(\game, \phi,\overline{V})$ \Comment{Eq. \eqref{qp_1}}
		\If {$V_{s,i+1} \geq 1-p$}
		\State \Return i+1 \Comment{breaking point is $i+1$}
		\EndIf
		\State $i++$
		
		\EndWhile
		\State \Return $\oslash$
	\end{algorithmic} 
\end{algorithm} 

\begin{restatable}{lemma}{theoremQPSol}
	\label{thm:qp-sol}
	\Cref{alg:k_resilience} terminates and returns $\max_{\pi}\breakpoint^T_{\game,\objective}$ if  $\max_{\pi}\breakpoint^T_{\game,\objective}\leq k$.
\end{restatable}

\begin{proof}

	Because of our assumptions, $\game^{\dagger k}$ is a stopping game, i.e., i) each state has only two available actions and ii) for every action, with a small probability $\epsilon$,	the game leads to $B$.
	
	Using the method of solving stopping SGs~\cite{condon1990algorithms},
	we can design a QP for $\game^{\dagger k}$ (hereafter referred to as $QP^{\dagger k}$). 
	Note that in this QP, we will minimize $F = \sum_i F_i$ w.r.t. $C=\bigcup_i C_i$
	(For $s\in\gamestates$, $0\leq i\leq k$, $a \in \gameavailactions(s)$, $V_{s,i}$  and $V_{s,i,a}$ are the variables corresponding to the states $(s,i)$ and $(s,i,a)$ respectively).
	Also, from \cite{condon1990algorithms}, $QP^{\dagger k}$ has a unique optimal solution $\overline{v}$. Also, for this optimal solution every summand in $F$ is $0$, in particular, $F_i$ is $0$ for all $i$.
	In this solution the value of $V_{s_0,k}$ is $\max_{\pi}\prob^{\pi}_{\game^{\dagger k}, (s_0,k)}(\phi) \geq p$.
	For $0\leq i\leq k$, we define $v_i$ as the projection of $\overline{v}$ corresponding to the variables $V_{s,i}$'s and $V_{s,i,a}$'s. 
	
	To proof the correctness of \cref{alg:k_resilience}, we will show that $\overline{v}$ is the unique solution of the iterative $QP$.
	Since $\overline{v}$ in the constraint set $C$, $v_i$ is in $C_i$.
	For any assignment satisfying $C_i$, every summand in $F_i$ is non-negative as they are products of two terms that are either both non-negative or both non-positive. Since, for the assignment $v$, $F_i=0$, $v_i$ is indeed an optimal solution. 
	To show the uniqueness, assume there is another solution $u\neq v$, such that $u_i$ is in $C_i$ and $F_i=0$. But then $u$ is in $C$ for which $F=0$, contradicting the fact that $QP^{\dagger k}$ has an unique solution.	
\end{proof}
This proves the following:
\begin{restatable}{lemma}{transWCBreakPoint}
	\label{thm:transWCBreakPoint}
	Given $k\in \N$ in unary, checking if $\max_{\pi}\breakpoint^T_{\game,\objective}\leq k$ is in $\polytime^{\NP}$.
\end{restatable}

\begin{proof}
	Note that the decision problem corresponding to these specific QPs defined above can be solved in $\NP$~\cite{condon1992complexity}, but in our case we also need to find the solution of the QP as we are using them in the next QP. 
	This can be extracted using polynomially many $\NP$ queries, where in $j^{th}$ query we will ask whether the $j^{th}$ bit of the solution is $1$. Also, in each iteration, we only need to remember the solution of the last iteration. 
	Thus, we can solve this in polynomial time with an $\NP$ oracle if $k$ is given in unary.
\end{proof}
\subsubsection{Iterative QP in the General Case}

Here we will show how to generalize these restrictions:
\paragraph*{Dealing with A1 and A2.}
If there are more than two actions at a state, the objective functions are not quadratic, 
but we can make them quadratic by adding intermediary states.

For example, suppose from a Player 1 state $s$, we have $l$ actions $a_0,a_1,\ldots a_{l-1}$ and $m$ deviation actions $d_0,d_1,\ldots d_{m-1}$.
First, from state $s$, we only allow $d_0$ and have two new actions that will lead to a new intermediary state $s_1$. Again from this state $s_1$, we will allow $d_1$ and have two new actions that will lead to a new intermediary state $s_2$. Once, we cover all deviation actions using $m$ states, we add additional states for the actions at $s$ to cover the actions. Two actions are taken and given to a new vertex $s_{m+1}$. State $s_m$ has then one action leading to an additional state $s_{m+1}$
instead of its previous two actions. This can be done iteratively until there is a binary tree with $s_m$ the root and has only
two actions. 
For this  $|\gameavailactions(s)|+|\gameavaildisturbanceactions(s)|-2$ more states for every state.

\paragraph*{Dealing with A3.}
Any arbitrary stochastic game can be converted to a polynomially larger stopping game using the construction in \cite{shapley1953stochastic}.
The idea is to add a transition with a small probability $\epsilon$
leading to a non-target sink-state for to every action. 
If the $\epsilon$ is chosen sufficiently small, one can infer the value in the original
SG from the modified SG~\cite[Lemma 8]{condon1992complexity}.

\theoremFreqResNP*
\begin{proof}
	If the probability of reaching $B$ is $\geq 1-p$, then the frequency breaking point is $0$ and the transient breaking point is computed using \cref{alg:k_resilience} whose correctness follows from \cref{thm:transWCBreakPoint}.
	If the probability of reaching $R\cup B$ is $< 1-p$, then it is not possible to break the strategy because the probability to reach $G$ under any strategies is $>p$.
	
	If the probability of reaching $B$ is $< 1-p$, and the probability of reaching $R\cup B$ is $\geq 1-p$, then the transient breaking point is infinite since staying outside of $G$ required staying is some other end component and just staying in end components with zero cost ($B$) is not enough.
	Therefore, to break the strategy Player~2 has to force the play to reach other end components and stay there. 
	Now, staying in the end components that have minimal cost is the same as the breaking point.
	
\end{proof}

%%%%%%%%%%%%%%%%%%%%%%%%%%%%%%%%%%%%%%%%%%%%%%%%%%%%%%%%%%%%%%%%%%%%%%%%
% \clearpage
% \input{changes.tex}
% \clearpage
% \input{relevance.tex}

%%%%%%%%%%%%%%%%%%%%%%%%%%%%%%%%%%%%%%%%%%%%%%%%%%%%%%%%%%%%%%%%%%%%%%%%

\end{document}